\documentclass[sn-mathphys-num,pdflatex]{sn-jnl}

\usepackage{fix-cm}
\usepackage{graphicx}
\usepackage{multirow}
\usepackage{amsmath,amssymb,amsfonts}
\usepackage{amsthm}
\usepackage{mathrsfs}
\usepackage[title]{appendix}
\usepackage{xcolor}
\usepackage{textcomp}
\usepackage{manyfoot}
\usepackage{booktabs}
\usepackage{listings}
\usepackage{tabularray}
\usepackage{overpic}
\usepackage{mathtools}
\usepackage{comment}

\renewcommand{\le}{\leqslant}

\renewcommand{\ge}{\geqslant}

\newcommand{\cC}{\mathcal{C}}

\newcommand{\cS}{\mathcal{S}}

\newcommand{\Z}{\mathbb{Z}}

\newcommand{\eqdef}{\triangleq}

\theoremstyle{thmstyleone}
\newtheorem{theorem}{Theorem}
\newtheorem{proposition}[theorem]{Proposition}
\newtheorem{corollary}[theorem]{Corollary}
\newtheorem{lemma}[theorem]{Lemma}

\theoremstyle{thmstyletwo}

\newtheorem{remark}{Remark}

\theoremstyle{thmstylethree}
\newtheorem{definition}{Definition}
\hypersetup{bookmarksdepth=subsubsection}
\DeclareMathOperator{\supp}{supp}

\newif\ifshowproofs
\showproofstrue 

\begin{document}

\title[Polynomial Constructions and Deletion-Ball Geometry for Multiset Deletion Codes
]{Polynomial Constructions and Deletion-Ball Geometry for Multiset Deletion Codes
}

\author[1]{\fnm{Avraham} \sur{Kreindel}}
\email{avrahamkreindel@gmail.com}

\author[2]{\fnm{Isaac} \sur{Barouch Essayag}}
\email{isaac.es@migal.org.il}

\author[3]{\fnm{Aryeh~Lev} \sur{Zabokritskiy (Yohananov)}}
\email{yuhanalev@telhai.ac.il}

\affil[1]{%
  \orgdiv{Department of Computer Science},
  \orgname{Reichman University},
  \orgaddress{\city{Herzliya}, \country{Israel}}
}

\affil[2]{%
  \orgdiv{Research Assistant},
  \orgname{MIGAL -- Galilee Research Institute/Tel-Hai University of Kiryat Shmona and the Galilee},
  \orgaddress{\city{Kiryat Shmona}, \country{Israel}}
}

\affil[3]{%
  \orgdiv{Department of Computer Science},
  \orgname{MIGAL -- Galilee Research Institute/Tel-Hai University of Kiryat Shmona and the Galilee},
  \orgaddress{\city{Kiryat Shmona}, \country{Israel}}
}

\abstract{
We study error-correcting codes in the space $\cS_{n,q}$ of length-$n$ multisets over a $q$-ary alphabet under the deletion metric, motivated by permutation channels in which ordering is completely lost and errors act only on symbol multiplicities. We develop two complementary directions. First, we present polynomial Sidon-type constructions over finite fields, in both projective and affine forms, yielding multiset $t$-deletion-correcting codes in the regime $t<q$ with redundancy $t+O(1)$, independent of the blocklength $n$. Second, we develop a geometric analysis of deletion balls in $\cS_{n,q}$. Using difference-vector representations together with a diagonal reduction of the relevant generating functions, we derive exact generating-function expressions for individual deletion-ball sizes, exact formulas for the number of ordered pairs of multisets at a fixed distance $m$, and consequently for the average ball size. We prove that radius-$r$ deletion balls are minimized at extreme multisets and maximized at the most balanced multisets, giving a formal global characterization of extremal centers in $\cS_{n,q}$. We further relate the maximal-ball value to the ideal difference set $S_{q-1}(r,r)$ through boundary truncation, obtaining explicit closed forms for $q=2$ and $q=3$. These geometric results lead to volume-based bounds on code size, including sphere-packing upper bounds, a boundary-aware analysis of code--anticode arguments, and Gilbert--Varshamov-type lower bounds governed by exact average ball sizes. For fixed $q$ and $t$, the resulting average-ball lower bound matches the interior-difference-set scale asymptotically.
}



\maketitle

\section{Introduction}
\label{sec:intro}

Communication models in which the order of transmitted symbols is unreliable or
irrelevant arise naturally in a variety of modern systems. Unlike classical
sequence-based channels, these models preserve only the \emph{multiset} of
symbols, discarding positional information entirely. This abstraction was
introduced and systematically developed in a series of works by
Kova\v{c}evi\'{c} and collaborators
\cite{KovacevicVukobratovic2013,KovacevicTan2018,KovacevicDuplication2019},
who showed that many impairments of permutation channels may be expressed as
operations on multiplicity vectors inside a discrete simplex. In such settings,
errors do not alter symbol positions but instead modify their counts, leading to
coding problems that are structurally different from classical Hamming or
sequence-based insertion--deletion models.

A key motivation, emphasized in~\cite{KovacevicTan2018}, is the connection to
\emph{permutation channels}, where the transmitted sequence may undergo arbitrary
reordering before reception. Similar abstractions arise in molecular communication,
DNA-based storage, and other settings in which the receiver observes an unordered
collection of symbols rather than an indexed sequence. In all of these settings,
deletions manifest as losses in symbol multiplicities, so the natural ambient
space is the multiset simplex $\cS_{n,q}$ rather than the classical Hamming cube.

In our previous work~\cite{CompanionPaper2026}, we focused on the extremal
high-deletion regime, exact results for binary and near-extremal parameters, and
explicit cyclic Sidon-type constructions. The present paper develops two
additional directions that were not the main focus there.

First, we study \emph{polynomial Sidon-type constructions}. These constructions
use factorization in finite polynomial rings to define algebraic syndrome maps,
and they are well suited to the regime of fixed alphabet size and
small deletion radius. In contrast with cyclic constructions, their redundancy is
essentially linear in the number of deletions and independent of the multiset
size $n$.

Second, we study the \emph{geometry of deletion balls}. Unlike the Hamming metric,
balls of a fixed radius in $\cS_{n,q}$ are non-homogeneous: their size
depends strongly on the center. This non-uniformity is the main obstacle in
transferring standard volume arguments to the multiset setting. We therefore
develop a generating-function framework for exact ball enumeration, use it to
prove a global extremal characterization of deletion-ball geometry, and then
derive upper and lower bounds governed by minimal, maximal, and average ball
sizes. In particular, we determine which centers minimize and which centers
maximize the size of a radius-$r$ deletion ball.

A relevant comparison point is the recent work of Goyal, Dao, Kova\v{c}evi\'{c},
and Kiah~\cite{GoyalDaoKovacevicKiah2025}, which studies Gilbert--Varshamov
bounds in several $L_1$ spaces, including the simplex, via multivariate analytic
combinatorics. Their focus is asymptotic and is tailored to scaling regimes in
which the alphabet size grows with the blocklength. In contrast, our goal here
is to obtain exact finite-$n$ formulas in the fixed-alphabet multiset model.

More precisely, we study three related enumerative objects. The first is
$A_{n,q}(m)$, the number of \emph{ordered} pairs of multisets in $\cS_{n,q}$ at
distance exactly $m$, which captures the exact distance distribution of the
space. The second is the average ball size $\overline{B_r}(n,q)$, obtained by
averaging the radius-$r$ deletion-ball size over all centers. The third is the
quantity $|S_{q-1}(r,r)|$, the size of the ideal difference set that
controls the maximal-ball phenomenon away from boundary truncation.

Our main generating-function theorem gives explicit closed formulas for these
quantities in the fixed-alphabet regime. In particular, it yields exact
finite-$n$ formulas for $A_{n,q}(m)$ and $\overline{B_r}(n,q)$, and it also
recovers the known formula for $|S_{q-1}(r,r)|$ within the same one-variable
framework. Thus the exact distance distribution, the exact average-ball size,
and the interior anticode size are placed within one common enumerative
picture. Since the formula for $|S_{q-1}(r,r)|$ already appears in
Kova\v{c}evi\'{c} and Tan~\cite[Lemma~1]{KovacevicTan2018}, the novelty here
lies in the unified derivation of this quantity together with the exact pair and
average-ball counts.

\subsection*{Contributions}
Our main contributions in this paper are as follows.

\begin{itemize}
\item We present projective and affine polynomial Sidon-type constructions for
multiset deletion-correcting codes, prove correctness via uniqueness of
factorization modulo a polynomial, and obtain redundancy $t+O(1)$ in the
regime $t<q$.

\item We derive a generating-function formula for the size of a radius-$r$
deletion ball around an arbitrary center, using difference vectors and a diagonal
reduction of the relevant bivariate series.

\item We obtain an exact rational generating function for the pair enumerator
of the metric space, yielding explicit formulas for the number $A_{n,q}(m)$ of
ordered pairs of multisets at distance $m$ and consequently for the exact
average ball size.

\item We prove a global extremal characterization of deletion-ball sizes:
minimal balls occur at extreme multisets, while maximal balls occur at the most
balanced multisets. We further show that the maximal value is governed by the
ideal difference set $S_{q-1}(r,r)$, with equality in the interior and
boundary truncation near the boundary. For $q=2$ and $q=3$ we obtain fully
explicit specializations.

\item We apply these results to derive volume-based bounds on code size,
including sphere-packing upper bounds, a boundary-aware interpretation of
code--anticode arguments, and Gilbert--Varshamov lower bounds based on average
balls.
\end{itemize}

The rest of the paper is organized as follows. Section~\ref{sec:pre} introduces
the multiset model and the linear/Sidon viewpoint. Section~\ref{sec:polynomial}
presents the polynomial constructions. Section~\ref{sec:balls} develops the
generating-function machinery and the geometry of deletion balls. Section~\ref{sec:vol-bounds}
derives the corresponding volume-based bounds. We conclude in
Section~\ref{sec:conclusion}.

\section{Preliminaries}
\label{sec:pre}

We briefly recall the notation and basic definitions used throughout the paper.
For additional background on linear multiset codes and Sidon-type constructions,
we refer the reader to our previous work~\cite{CompanionPaper2026}.

Let
\[
   \Sigma=\{0,1,\dots,q-1\}, \qquad |\Sigma|=q,
\]
and let $\cS_{n,q}$ denote the set of all multisets of cardinality $n$ over
$\Sigma$. Equivalently, $\cS_{n,q}$ may be identified with the discrete simplex
\[
   \cS_{n,q}
   =
   \Bigl\{
      (x_0,\dots,x_{q-1})\in\mathbb{Z}_{\ge 0}^q
      : \sum_{i=0}^{q-1} x_i=n
   \Bigr\},
\]
so that
\[
   |\cS_{n,q}|=\binom{n+q-1}{q-1}.
\]

For a multiset $S\in\cS_{n,q}$ we write $S(i)$ for the multiplicity of the
symbol $i\in\Sigma$, and we denote its multiplicity vector by
\[
   \mathbf{x}_S=(S(0),S(1),\dots,S(q-1)).
\]

\begin{definition}[Basic multiset operations]
\label{def:multiset-ops}
For multisets $S,T$ over $\Sigma$, define
\[
   (S\cap T)(i)=\min\{S(i),T(i)\},
   \qquad
   (S\setminus T)(i)=\max\{S(i)-T(i),0\},
\]
and
\[
   (S\uplus T)(i)=S(i)+T(i),
   \qquad i\in\Sigma.
\]
\end{definition}

A \emph{code} is a subset $\cC\subseteq\cS_{n,q}$. A \emph{deletion} removes one
occurrence of some symbol, and hence decreases the total cardinality by one.

\begin{definition}[Multiset deletion code]
\label{def:multiset-code}
A code $\cC\subseteq\cS_{n,q}$ of size $M$ that corrects any $t$ deletions is
called an \emph{$S_q[n,M,t]$ multiset $t$-deletion-correcting code}.  For
$q=2$ we write $S[n,M,t]$.  The maximum possible size of such a code is denoted
by $S_q(n,t)$.
\end{definition}

We use the deletion metric introduced in~\cite{KovacevicTan2018}:
\[
   d(S,T)=n-|S\cap T|,
   \qquad S,T\in\cS_{n,q}.
\]
In terms of multiplicity vectors, this is equivalent to the scaled $\ell_1$
metric
\[
   d(S,T)
   =
   \frac12 \sum_{i=0}^{q-1} \bigl|S(i)-T(i)\bigr|
   =
   \frac12 \|\mathbf{x}_S-\mathbf{x}_T\|_1.
\]
Accordingly, the minimum distance of a code $\cC$ is
\[
   d(\cC)=\min_{\substack{S,T\in\cC\\ S\neq T}} d(S,T).
\]

As in the previous paper, a code corrects $t$ deletions if and only if its
minimum distance is at least $t+1$.

\medskip
\noindent
\textbf{Binary specialization.}
When $q=2$, every multiset $S\in\cS_{n,2}$ is uniquely determined by its weight
\[
   w(S)\eqdef S(1),
\]
since then
\[
   \mathbf{x}_S=(n-w(S),\,w(S)).
\]
In this case the metric becomes one-dimensional:
\[
   d(S,T)=|w(S)-w(T)|.
\]
Thus the binary multiset space is isometric to the integer interval
$\{0,1,\dots,n\}$ with the usual absolute-value metric. We will repeatedly use
this simplification later in the paper.

For $r\ge 0$ and $S\in\cS_{n,q}$, the radius-$r$ deletion ball centered at $S$
is
\[
   B_r(S)=\{T\in\cS_{n,q}: d(S,T)\le r\}.
\]
Since the space $\cS_{n,q}$ is not homogeneous, the size of $B_r(S)$ generally
depends on the center $S$. We therefore define
\[
   L_r \eqdef \max_{S\in\cS_{n,q}} |B_r(S)|,
   \qquad
   M_r \eqdef \min_{S\in\cS_{n,q}} |B_r(S)|,
\]
and the average radius-$r$ ball size
\[
   \overline{B_r}(n,q)
   \eqdef
   \frac{1}{|\cS_{n,q}|}\sum_{S\in\cS_{n,q}} |B_r(S)|.
\]
For $m\ge 0$, we also write
\[
   A_{n,q}(m)
   \eqdef
   \bigl|\{(S,T)\in\cS_{n,q}^2:\ d(S,T)=m\}\bigr|
\]
for the number of \emph{ordered} pairs of multisets at distance exactly $m$.

It will also be convenient to use the universal difference-vector set
\[
   S_{q-1}(r^{+},r^{-})
   \eqdef
   \Bigl\{
        \mathbf{z}\in\mathbb{Z}^{q}:
        \sum_{i=0}^{q-1} z_i = 0,\;
        \sum_{i=0}^{q-1} z_i^{+} \le r^{+},\;
        \sum_{i=0}^{q-1} z_i^{-} \le r^{-}
   \Bigr\},
\]
which plays the role of an ideal interior ball in difference-vector
coordinates.

Finally, we recall the notion of linearity used throughout the multiset-code
literature.

Let
\[
   A^{q-1}
   =
   \Bigl\{
      \mathbf{z}\in\mathbb{Z}^q :
      \sum_{i=0}^{q-1} z_i = 0
   \Bigr\}.
\]

\begin{definition}[Linear multiset code]
\label{def:linearity}
A code $\cC\subseteq\cS_{n,q}$ is called \emph{linear} if there exist a lattice
$L\subseteq A^{q-1}$ and a vector $\mathbf{t}\in\mathbb{Z}^q$ with
$\sum_{i=0}^{q-1} t_i=n$ such that
\[
   \{\mathbf{x}_S : S\in\cC\}
   =
   (L+\mathbf{t})\cap \mathbb{Z}_{\ge 0}^q.
\]
\end{definition}

We will also use the standard Sidon-type language.

\begin{definition}[Generalized Sidon set]
\label{def:Bt}
Let $G$ be an abelian group, and let
\[
   B=\{b_0,b_1,\dots,b_{m-1}\}\subseteq G.
\]
For $t\ge 1$, we say that $B$ is a \emph{$B_t$-set} if all multiset sums of at
most $t$ elements of $B$ are distinct.
\end{definition}

The relation between linear multiset codes and Sidon-type sets will be used in
Section~\ref{sec:polynomial}; see also~\cite{KovacevicTan2018,CompanionPaper2026}
for a broader discussion.

\section{Polynomial Sidon-Type Constructions}
\label{sec:polynomial}

In this section we present a polynomial-based Sidon-type construction of
multiset deletion-correcting codes.
Compared with the cyclic Sidon-type approach developed in our previous
work~\cite{CompanionPaper2026}, the main advantage here is a clean algebraic
description in terms of factorization modulo a polynomial together with an
explicit syndrome-decoding procedure. In the fixed-alphabet regime it yields
redundancy linear in~$t$, independent of the blocklength $n$.

For clarity, throughout this section we use the letter $s$ for the size of the
underlying finite field. Thus, in the \emph{projective} version the alphabet has
size $s+1$, whereas in the \emph{affine} version the alphabet has size exactly
$s$.

\subsection{Projective polynomial Sidon-type construction}

Let $s$ be a prime power and let
\[
A \;=\; \mathbb{P}^1(\mathbb{F}_s) \;=\; \mathbb{F}_s \cup \{\infty\},
\qquad |A|=s+1.
\]
Let $f\in\mathbb{F}_s[X]$ be a polynomial satisfying
\[
\deg(f)=t+1
\qquad\text{and}\qquad
f(a)\neq 0 \;\;\text{for all } a\in\mathbb{F}_s.
\]
For example, any irreducible polynomial of degree $t+1$ satisfies this
condition.

Define
\[
R_f \;\eqdef\; \mathbb{F}_s[X]/(f),
\qquad
G \;\eqdef\; R_f^{*}/\mathbb{F}_s^{*},
\]
where $R_f^{*}$ is the group of units of $R_f$.
Let \(\xi\in R_f\) denote the image of \(X\) modulo \(f\).

For \(g\in R_f^{*}\), write \([g]\) for its class in the quotient group \(G\):
\[
[g] \;\eqdef\; g\cdot \mathbb{F}_s^{*}
\;=\;
\{\lambda g:\lambda\in\mathbb{F}_s^{*}\}.
\]
The group operation in \(G\) is induced by multiplication in \(R_f^*\), namely
\[
[g][h]\;\eqdef\;[gh].
\]
This is well defined: if \(g'=\lambda g\) and \(h'=\mu h\) with
\(\lambda,\mu\in\mathbb{F}_s^*\), then
\[
g'h'=(\lambda\mu)gh,
\]
and hence \([g'h']=[gh]\).

Let \(\xi\in R_f\) denote the image of \(X\) modulo \(f\). Define a mapping \(\psi:A\to G\) by
\[
\psi(a)=
\begin{cases}
[\xi-a], & a\in\mathbb{F}_s,\\[2pt]
[1],     & a=\infty.
\end{cases}
\]
Thus \([\xi-a]\) is the class in \(G\) of the residue class of \(X-a\) modulo \(f\).
This induces a homomorphism on multisets
\[
\Psi:\cS_{n,s+1}\longrightarrow G,
\qquad
\Psi(S)=\prod_{a\in A}\psi(a)^{S(a)}.
\]

For a fixed syndrome class \(c\in G\), define
\[
\cC^{\mathrm{proj}}_{n,s+1}(c)
\;=\;
\{\,S\in\cS_{n,s+1} : \Psi(S)=c\,\}.
\]

\begin{theorem}
\label{thm:projective-polynomial}
For every \(c\in G\), the code \(\cC^{\mathrm{proj}}_{n,s+1}(c)\) is
\(t\)-deletion-correcting.
\end{theorem}

\begin{proof}
Suppose \(S_1,S_2\in\cC^{\mathrm{proj}}_{n,s+1}(c)\) and that, after at most \(t\)
deletions, both yield the same received multiset \(R\).
Since \(S_1\) and \(S_2\) both have cardinality \(n\), the same number
\(r\le t\) of symbols must have been deleted from each.
Hence there exist multisets \(E_1,E_2\) of cardinality \(r\) such that
\[
S_1=R\uplus E_1,
\qquad
S_2=R\uplus E_2.
\]
Using multiplicativity of \(\Psi\) and the fact that \(\Psi(S_1)=\Psi(S_2)=c\),
we obtain
\[
\Psi(E_1)=\Psi(E_2).
\]
Therefore, for some scalar \(C\in\mathbb{F}_s^{*}\),
\[
\prod_{a\in\mathbb{F}_s}(X-a)^{E_1(a)}
\;\equiv\;
C\prod_{a\in\mathbb{F}_s}(X-a)^{E_2(a)}
\pmod f.
\]
Set
\[
P_i(X)\eqdef \prod_{a\in\mathbb{F}_s}(X-a)^{E_i(a)},
\qquad i=1,2.
\]
Since deletions of the symbol \(\infty\) do not contribute to \(P_i\), one has
\[
\deg(P_i)=r-E_i(\infty)\le r\le t,
\qquad i=1,2.
\]
Therefore
\[
\deg(P_1-CP_2)\le \max\{\deg(P_1),\deg(P_2)\}\le t<\deg(f)=t+1.
\]
Since \(P_1\equiv CP_2\pmod f\), the polynomial \(P_1-CP_2\) is divisible by
\(f\). Because its degree is strictly smaller than \(\deg(f)\), it must vanish
identically:
\[
P_1=CP_2.
\]
In particular, \(P_1\) and \(P_2\) have the same degree, so
\(E_1(\infty)=E_2(\infty)\). Both polynomials are monic, so necessarily
\(C=1\), and thus \(P_1=P_2\).
By unique factorization in \(\mathbb{F}_s[X]\), we conclude that \(E_1=E_2\),
and therefore
\[
S_1=R\uplus E_1=R\uplus E_2=S_2.
\]
Thus the code corrects \(t\) deletions.
\end{proof}

\subsection{Affine polynomial variant}

To obtain a construction over an alphabet of size exactly \(s\), we use an
affine variant.
Let
\[
A=\mathbb{F}_s,
\qquad |A|=s,
\]
and let \(f\in\mathbb{F}_s[X]\) satisfy
\[
\deg(f)=t
\qquad\text{and}\qquad
f(a)\neq 0 \;\;\text{for all } a\in\mathbb{F}_s.
\]
Again, any irreducible polynomial of degree \(t\) satisfies this condition when
\(t\ge 2\).

Define
\[
R_f \;\eqdef\; \mathbb{F}_s[X]/(f),
\qquad
G \;\eqdef\; R_f^{*},
\]
and let \(\xi\in R_f\) denote the image of \(X\) modulo \(f\).
Define
\[
\psi(a)=\xi-a \in G,
\qquad a\in\mathbb{F}_s.
\]
This induces
\[
\Psi:\cS_{n,s}\longrightarrow G,
\qquad
\Psi(S)=\prod_{a\in\mathbb{F}_s}\psi(a)^{S(a)}.
\]
For a fixed \(c\in G\), define
\[
\cC^{\mathrm{aff}}_{n,s}(c)
\;=\;
\{\,S\in\cS_{n,s} : \Psi(S)=c\,\}.
\]

\begin{theorem}
\label{thm:affine-polynomial}
For every \(c\in G\), the code \(\cC^{\mathrm{aff}}_{n,s}(c)\) is
\(t\)-deletion-correcting.
\end{theorem}

\begin{proof}
Let \(S_1,S_2\in\cC^{\mathrm{aff}}_{n,s}(c)\) produce the same received multiset
\(R\) after \(r\le t\) deletions. Write
\[
S_1=R\uplus E_1,
\qquad
S_2=R\uplus E_2,
\]
where \(|E_1|=|E_2|=r\).
Then
\[
\Psi(E_1)=\Psi(E_2).
\]
Equivalently,
\[
\prod_{a\in\mathbb{F}_s}(X-a)^{E_1(a)}
\;\equiv\;
\prod_{a\in\mathbb{F}_s}(X-a)^{E_2(a)}
\pmod f.
\]
Define
\[
P_i(X)\eqdef \prod_{a\in\mathbb{F}_s}(X-a)^{E_i(a)},
\qquad i=1,2.
\]
Since \(P_1\) and \(P_2\) are monic of the same degree \(r\le t\), we have
\[
\deg(P_1-P_2)<r\le t=\deg(f).
\]
But \(P_1-P_2\) is divisible by \(f\), hence it must be the zero polynomial.
Therefore \(P_1=P_2\), and unique factorization gives \(E_1=E_2\).
Consequently \(S_1=S_2\), so the code corrects \(t\) deletions.
\end{proof}

The affine variant preserves the same syndrome-decoding philosophy as the
projective one, but works over an alphabet of size exactly \(s\), without the
auxiliary projective symbol \(\infty\).

\subsection{Syndrome decoding}

The algebraic form of the construction leads to a natural decoding rule.

\begin{proposition}[Syndrome decoding]
\label{prop:syndrome-decoding}
Consider either the projective construction
\[
\cC^{\mathrm{proj}}_{n,s+1}(c)
=
\{\,S\in\cS_{n,s+1}:\Psi(S)=c\,\}
\]
or the affine construction
\[
\cC^{\mathrm{aff}}_{n,s}(c)
=
\{\,S\in\cS_{n,s}:\Psi(S)=c\,\}.
\]
Let \(S\) be a transmitted codeword and suppose that the received multiset \(R\)
is obtained from \(S\) by deleting \(r\le t\) symbols. Then there exists a
unique multiset \(E\) of cardinality \(r\) such that
\[
S=R\uplus E
\]
and
\[
\Psi(E)=c\,\Psi(R)^{-1}.
\]
Consequently, decoding reduces to recovering the unique multiset \(E\) of size
exactly \(r=n-|R|\) whose syndrome equals \(c\,\Psi(R)^{-1}\), and then
outputting
\[
\widehat{S}=R\uplus E.
\]
\end{proposition}

\begin{proof}
Existence is immediate: if \(E\) is the multiset of deleted symbols, then
\[
S=R\uplus E,
\]
and by multiplicativity of \(\Psi\),
\[
\Psi(S)=\Psi(R)\Psi(E).
\]
Since \(S\) belongs to the syndrome class \(c\), we have \(\Psi(S)=c\), and
therefore
\[
\Psi(E)=c\,\Psi(R)^{-1}.
\]

For uniqueness, suppose that \(E_1\) and \(E_2\) are multisets of cardinality
\(r\le t\) such that
\[
\Psi(E_1)=\Psi(E_2).
\]
Applying exactly the same polynomial argument as in the proof of
Theorem~\ref{thm:projective-polynomial} in the projective case, and of
Theorem~\ref{thm:affine-polynomial} in the affine case, we conclude that
\(E_1=E_2\).
\end{proof}

\begin{remark}[Algorithmic implementation]
\label{rem:syndrome-decoding-algorithm}
For fixed \(s\) and \(t\), and for each \(0\le r\le t\), define
\[
\mathcal{E}_r
=
\{\,E : |E|=r\,\},
\]
where \(E\) ranges over multisets on the relevant alphabet
(\(\mathbb{P}^1(\mathbb{F}_s)\) in the projective case, or \(\mathbb{F}_s\) in
the affine case). By Proposition~\ref{prop:syndrome-decoding}, the map
\[
E\longmapsto \Psi(E)
\]
is injective on each \(\mathcal{E}_r\). Hence one may precompute the lookup
table
\[
\mathcal{T}_r
=
\{\,(\Psi(E),E): E\in\mathcal{E}_r\,\}.
\]

Given a received multiset \(R\), the decoder first determines
\[
r=n-|R|,
\]
then computes
\[
\sigma \eqdef c\,\Psi(R)^{-1},
\]
retrieves the unique \(E\in\mathcal{E}_r\) such that
\[
\Psi(E)=\sigma,
\]
and outputs
\[
\widehat{S}=R\uplus E.
\]

In the projective case,
\[
|\mathcal{E}_r|=\binom{r+s}{s},
\]
while in the affine case,
\[
|\mathcal{E}_r|=\binom{r+s-1}{s-1}.
\]
Thus, for fixed \(s\) and \(t\), preprocessing is finite and independent of
\(n\), and online decoding requires \(O(n)\) group operations plus a table
lookup.

Encoding may be viewed as enumerative encoding inside a fixed syndrome class.
Since the ambient multiset space has polynomial size in \(n\) for fixed alphabet
size, a straightforward encoder based on enumerating multiplicity vectors and
retaining only those with the prescribed syndrome runs in polynomial time in
\(n\). A more structured systematic encoder can also be designed by reserving a
short check part and solving a small syndrome-completion problem in the group
\(G\), but we do not pursue this here.
\end{remark}

\subsection{Worked example over \texorpdfstring{$\mathbb{F}_3$}{F3}}

We now illustrate the projective construction on a small explicit example.

\medskip
\noindent
\textbf{Example (projective construction over \(\mathbb{F}_3\)).}
We construct a code in
\[
\cS_{3,4},
\]
that is, a \(4\)-ary multiset code of cardinality \(n=3\), correcting
\(t=1\) deletion.
Let
\[
A=\mathbb{P}^1(\mathbb{F}_3)=\{0,1,2,\infty\}.
\]
Choose
\[
f(X)=X^2+1.
\]
Since
\[
f(0)=1,\qquad f(1)=2,\qquad f(2)=2,
\]
the polynomial \(f\) has no roots in \(\mathbb{F}_3\), and hence it is
irreducible.

Set
\[
R_f=\mathbb{F}_3[X]/(X^2+1),
\qquad
G=R_f^{*}/\mathbb{F}_3^{*},
\]
and let
\[
\alpha \;\eqdef\; X \bmod (X^2+1)\in R_f.
\]
Then
\[
\alpha^2=-1=2.
\]

Since every residue class modulo \(X^2+1\) has a unique representative of
degree at most \(1\), we may write
\[
R_f
=
\{\,a+b\alpha : a,b\in\mathbb{F}_3\,\}
=
\{\,0,1,2,\alpha,2\alpha,1+\alpha,1+2\alpha,2+\alpha,2+2\alpha\,\}.
\]
Because \(f\) is irreducible, \(R_f\) is a field with \(3^2=9\) elements, so
\[
R_f^*=R_f\setminus\{0\}
=
\{\,1,2,\alpha,2\alpha,1+\alpha,1+2\alpha,2+\alpha,2+2\alpha\,\}.
\]

Moreover,
\[
\mathbb{F}_3^*=\{1,2\}.
\]
Hence the quotient group
\[
G=R_f^*/\mathbb{F}_3^*
\]
consists of the equivalence classes under multiplication by \(1\) and \(2\).
Explicitly,
\[
[1]=\{1,2\},
\qquad
[\alpha]=\{\alpha,2\alpha\},
\]
\[
[1+\alpha]=\{1+\alpha,2+2\alpha\},
\qquad
[1+2\alpha]=\{1+2\alpha,2+\alpha\}.
\]
Thus
\[
G=\bigl\{[1],[\alpha],[1+\alpha],[1+2\alpha]\bigr\}.
\]
Note that
\[
\alpha-1=\alpha+2=2+\alpha,
\]
so
\[
[\alpha-1]=[2+\alpha]=[1+2\alpha].
\]

The symbol map is therefore
\[
\psi(0)=[\alpha],\qquad
\psi(1)=[\alpha-1]=[1+2\alpha],\qquad
\psi(2)=[\alpha+1]=[1+\alpha],\qquad
\psi(\infty)=[1].
\]

Consider the syndrome class \([1]\in G\), and define
\[
\cC^{\mathrm{proj}}_{3,4}([1])
=
\{\,S\in\cS_{3,4}:\Psi(S)=[1]\,\}.
\]
A direct calculation shows that this code consists of the five multisets
\[
\{0,0,\infty\},\qquad
\{0,1,1\},\qquad
\{0,2,2\},\qquad
\{1,2,\infty\},\qquad
\{\infty,\infty,\infty\}.
\]

For example,
\[
\Psi(\{0,1,1\})
=
[\alpha]\,[\alpha-1]^2.
\]
Now
\[
(\alpha-1)^2=\alpha^2-2\alpha+1=2+\alpha+1=\alpha,
\]
because \(-2\equiv 1\pmod 3\). Hence
\[
\Psi(\{0,1,1\})
=
[\alpha]\,[\alpha]
=
[\alpha^2]
=
[2]
=
[1],
\]
so indeed
\[
\{0,1,1\}\in\cC^{\mathrm{proj}}_{3,4}([1]).
\]

To illustrate decoding, suppose that
\[
S=\{0,1,1\}
\]
is transmitted and one symbol is deleted, yielding
\[
R=\{0,1\}.
\]
Then \(r=1\), so the deleted multiset \(E\) has cardinality \(1\). Since
\[
\Psi(E)=\Psi(S)\Psi(R)^{-1}=[1]\Psi(R)^{-1},
\]
we compute
\[
\Psi(R)=\psi(0)\psi(1)=[\alpha][\alpha-1]=[\alpha(\alpha-1)].
\]
Therefore
\[
\Psi(E)=[\alpha(\alpha-1)]^{-1}.
\]
But
\[
\alpha(\alpha-1)(\alpha-1)=\alpha(\alpha-1)^2=\alpha\cdot\alpha=\alpha^2=2,
\]
and since \([2]=[1]\), it follows that
\[
[\alpha(\alpha-1)]^{-1}=[\alpha-1].
\]
Hence
\[
\Psi(E)=[\alpha-1]=\psi(1),
\]
so
\[
E=\{1\}.
\]
Thus the original codeword is recovered as
\[
S=R\uplus E=\{0,1\}\uplus\{1\}=\{0,1,1\}.
\]

This example illustrates the general principle: a codeword is simply a multiset
whose multiplicity vector satisfies one algebraic syndrome constraint in a
finite abelian group.

\textbf{Example (affine construction over \(\mathbb{F}_3\)).}\label{ex:tiny-affine}
Take \(s=3\) and \(t=2\), and choose
\[
f(X)=X^2+1,
\]
which has no roots in \(\mathbb{F}_3\). Then
\[
A=\mathbb{F}_3=\{0,1,2\},
\qquad
R_f=\mathbb{F}_3[X]/(X^2+1),
\qquad
G=R_f^*.
\]
If \(\xi\) denotes the image of \(X\) in \(R_f\), then
\[
\psi(0)=\xi,\qquad
\psi(1)=\xi-1,\qquad
\psi(2)=\xi+1.
\]
For any fixed \(c\in G\) and any blocklength \(n\), the code
\[
\cC^{\mathrm{aff}}_{n,3}(c)
=
\{\,S\in\cS_{n,3}:\Psi(S)=c\,\}
\]
consists of all ternary multisets of size \(n\) whose syndrome equals \(c\).
This is the same mechanism as in the projective case, but without the extra
symbol \(\infty\) and without quotienting by \(\mathbb{F}_3^*\).

\subsection{Size and redundancy}

The polynomial constructions partition the ambient multiset space into syndrome
classes indexed by the corresponding group \(G\), so at least one class has size
at least the average.

\medskip
\noindent
\textbf{Projective case.}
Here
\[
|R_f|=s^{t+1},
\]
and hence
\[
|G|=\frac{|R_f^*|}{|\mathbb{F}_s^*|}
\le \frac{|R_f|}{s-1}
=
\frac{s^{t+1}}{s-1}.
\]
Therefore, for some \(c\in G\),
\[
\bigl|\cC^{\mathrm{proj}}_{n,s+1}(c)\bigr|
\;\ge\;
\frac{|\cS_{n,s+1}|}{|G|}
\;\ge\;
\frac{s-1}{s^{t+1}}\binom{n+s}{s}
=
\Omega\!\left(\frac{\binom{n+s}{s}}{s^t}\right).
\]

\medskip
\noindent
\textbf{Affine case.}
Here
\[
|G|=|R_f^*|\le |R_f|=s^t.
\]
Hence, for some \(c\in G\),
\[
\bigl|\cC^{\mathrm{aff}}_{n,s}(c)\bigr|
\;\ge\;
\frac{|\cS_{n,s}|}{|G|}
\;\ge\;
\frac{1}{s^t}\binom{n+s-1}{s-1}.
\]

In either version, the redundancy is linear in \(t\) and independent of \(n\).
More precisely, for fixed \(s\) and \(t\),
\[
R(\cC)=t+O(1).
\]
Indeed, in the affine case one has
\[
R(\cC)\le \log_s |G|\le \log_s(s^t)=t.
\]
In the projective case, since \(q=s+1\),
\[
R(\cC)\le \log_{s+1}|G|
\le \log_{s+1}\!\left(\frac{s^{t+1}}{s-1}\right)
=(t+1)\log_{s+1}s-\log_{s+1}(s-1).
\]
Moreover,
\[
0\le (t+1)\bigl(1-\log_{s+1}s\bigr)
=(t+1)\log_{s+1}\!\left(1+\frac1s\right)
=\frac{(t+1)\ln(1+1/s)}{\ln(s+1)}.
\]
Since \(t<q=s+1\), one has \(t+1\le s+1\), and therefore
\[
\frac{(t+1)\ln(1+1/s)}{\ln(s+1)}
\le
\frac{(s+1)\ln(1+1/s)}{\ln(s+1)}
\le
\frac{s+1}{s\,\ln(s+1)}
=O(1).
\]
Hence also in the projective case one gets \(R(\cC)=t+O(1)\).

\begin{remark}[Sharper group size for irreducible \(f\)]
If \(f\) is irreducible, then \(R_f\) is a field. In the affine case this gives
\[
|G|=|R_f^*|=s^t-1,
\]
while in the projective case
\[
|G|=\frac{|R_f^*|}{|\mathbb{F}_s^*|}
=\frac{s^{t+1}-1}{s-1}.
\]
Hence one may replace the crude bounds above by the slightly sharper estimates
\[
\bigl|\cC^{\mathrm{aff}}_{n,s}(c)\bigr|
\;\ge\;
\frac{|\cS_{n,s}|}{s^t-1},
\qquad
\bigl|\cC^{\mathrm{proj}}_{n,s+1}(c)\bigr|
\;\ge\;
\frac{(s-1)\,|\cS_{n,s+1}|}{s^{t+1}-1}.
\]
\end{remark}

\subsection{Comparison with previous Sidon-type constructions}

It is instructive to compare the present polynomial construction with the cyclic
Sidon-type construction developed in our previous work~\cite{CompanionPaper2026}.
The two approaches are based on different algebraic principles and are strongest
in different parameter regimes.

\begin{itemize}
\item \emph{Advantages of the present polynomial approach:}
one advantage is smaller redundancy. For fixed small deletion radius
\(t\) and moderate or large alphabet size \(q\), the present construction
achieves redundancy of order \(t+O(1)\), whereas the redundancy of the cyclic
family in~\cite{CompanionPaper2026} depends explicitly on \(q\). Thus, although
the cyclic construction is more uniform across the full parameter range, the
present one can yield smaller redundancy when \(t\) is small and the alphabet
size is moderate or large.

A second advantage is structural: the present method yields a distinct
algebraic family of multiset deletion-correcting codes, rooted in finite-field,
affine, and projective ideas. In this sense, it is not merely a reformulation
of the cyclic approach of~\cite{CompanionPaper2026}, but an alternative
construction principle. This algebraic viewpoint may also be of independent
combinatorial interest.

\item \emph{Advantages of the cyclic approach of~\cite{CompanionPaper2026}:}
the cyclic construction applies uniformly to arbitrary alphabet sizes and
arbitrary deletion radii, and it is also algorithmically simpler. In addition,
the earlier paper establishes several exact results in extremal regimes that are
not the focus of the present polynomial method.

\item \emph{Limitations of the present polynomial approach:}
the alphabet size is constrained by finite-field arithmetic. In the affine
version one has \(q=s\), while in the projective version one has \(q=s+1\),
where \(s\) is a prime power. Moreover, the present construction requires
\(t<q\). Thus, while it is especially attractive for small deletion radius over
moderate or large alphabets, it is less flexible than the cyclic approach in the
full \((q,t)\)-range.
\end{itemize}

Accordingly, the two constructions should be viewed as complementary rather than
as one uniformly dominating the other: the present polynomial method is
well suited to the small-\(t\), moderate/large-\(q\) regime, while
the cyclic method of~\cite{CompanionPaper2026} remains more general and simpler
algorithmically.


\section{Balls in the Multiset Space}
\label{sec:balls}

In this section we study the structure and cardinalities of metric balls in the
multiset space $\cS_{n,q}$ under the deletion distance
\[
   d(S,T) = n - |S \cap T|.
\]
Unlike the Hamming metric, where all balls of a given radius have the same
size, balls in $\cS_{n,q}$ are \emph{non-uniform}: their size depends
strongly on the choice of the center.
This lack of regularity is a defining feature of the multiset model and has
direct consequences for both upper and lower bounds on the size of multiset
deletion codes.

Our primary focus is on the \emph{minimal} radius-$r$ ball size, as it governs
sphere-packing upper bounds.
Understanding which centers yield the smallest balls is therefore essential
for deriving meaningful limitations on multiset deletion codes.
In addition, we analyze maximal and average ball sizes in specific regimes.
While maximal balls are relevant for anticode-based arguments, average ball
sizes play a central role in Gilbert--Varshamov type lower bounds.

We begin by developing a general counting framework for radius-$r$ balls around
an arbitrary center.
This framework is based on a stars-and-bars interpretation of multisets,
difference vectors, and generating functions.
Although it applies to individual balls, it will be particularly useful later
for computing \emph{average} ball sizes, where summation over all centers
naturally interacts with the generating-function formulation.
The framework will then be specialized to identify extremal centers and to
derive explicit formulas in low-alphabet cases.

\subsection{A generating-function formula via difference vectors}
\label{subsec:ball-general-formula}

We adopt the classical \emph{stars-and-bars} interpretation of multisets.
A multiset $S\in\cS_{n,q}$ is viewed as a distribution of $n$ identical balls
(stars) among $q$ bins (symbols), where the multiplicity $x_i$ records the number
of balls in bin~$i$.
Under this interpretation, moving within a deletion ball corresponds to
redistributing balls among bins while preserving the total number of balls.

Fix a multiset $S\in\cS_{n,q}$ with multiplicity vector
\[
   \mathbf{x}_S = (x_0,\dots,x_{q-1}) \in \mathbb{Z}_{\ge 0}^q,
   \qquad \sum_{i=0}^{q-1} x_i = n .
\]
Here and throughout, bold symbols denote vectors, while non-bold symbols denote
their individual coordinates.
For any $T\in\cS_{n,q}$ with multiplicity vector
$\mathbf{x}_T=(y_0,\dots,y_{q-1})$, the deletion distance satisfies
\[
   d(S,T)
   = \frac12 \sum_{i=0}^{q-1} |x_i-y_i|.
\]

It is convenient to describe the radius-$r$ ball around $S$ in terms of
\emph{difference vectors}.
For each $T\in\cS_{n,q}$ define
\[
   \mathbf{z} = \mathbf{x}_T - \mathbf{x}_S \in \mathbb{Z}^q .
\]
Since both $S$ and $T$ contain exactly $n$ balls, any such vector satisfies
\[
   \sum_{i=0}^{q-1} z_i = 0 .
\]
Intuitively, a difference vector $\mathbf{z}$ records how balls are transferred
between bins: positive coordinates correspond to adding balls to certain bins,
while negative coordinates correspond to removing balls from others.
Preservation of the total number of balls forces these changes to balance.

For example, if $S=(2,1,1)$ and $T=(1,2,1)$, then
$\mathbf{z}=\mathbf{x}_T-\mathbf{x}_S=(-1,1,0)$, whose coordinates indeed sum
to zero.

Write $\mathbf{z}=\mathbf{z}^{+}-\mathbf{z}^{-}$, where
$\mathbf{z}^{+},\mathbf{z}^{-}\in\mathbb{Z}_{\ge 0}^q$ have disjoint supports and
\[
   z_i^{+}=\max\{z_i,0\}, \qquad
   z_i^{-}=\max\{-z_i,0\}.
\]
Here $\mathbf{z}^{+}$ represents the number of balls added to each bin, while
$\mathbf{z}^{-}$ represents the number of balls removed.
Balls may be added without restriction, whereas removals are limited by the
available multiplicities in $S$.
With this notation,
\[
   d(S,T)
   = \sum_{i=0}^{q-1} z_i^{+}
   = \sum_{i=0}^{q-1} z_i^{-}.
\]

Moreover, the nonnegativity constraint $\mathbf{x}_T\in\mathbb{Z}_{\ge 0}^q$
is equivalent to the componentwise inequality
\[
   \mathbf{z}^{-} \le \mathbf{x}_S .
\]

Consequently, the radius-$r$ ball around $S$ may be written as
\[
   B_r(S)
   =
   \Bigl\{
      \mathbf{z}\in\mathbb{Z}^q :
      \sum_{i=0}^{q-1} z_i = 0,\;
      \sum_{i=0}^{q-1} z_i^{+} \le r,\;
      \mathbf{z}^{-} \le \mathbf{x}_S
   \Bigr\}.
\]

\medskip
We encode these constraints using a bivariate generating function.
Introduce formal variables $s$ and $\tau$, where the exponent of $s$ tracks the
total number of added balls (i.e., $\sum_i z_i^{+}$), and the exponent of $\tau$
tracks the total number of removed balls (i.e., $\sum_i z_i^{-}$).

For a fixed coordinate $i$, the possible values of $z_i$ are:
\begin{itemize}
  \item $z_i = 0$,
  \item $z_i = k \ge 1$, contributing $s^k$,
  \item $z_i = -k$ for $1 \le k \le x_i$, contributing $\tau^k$.
\end{itemize}
Thus the single-coordinate generating function is
\[
   F_i(s,\tau)
   = 1 + \sum_{k\ge 1} s^k + \sum_{k=1}^{x_i} \tau^k .
\]

The coordinates contribute independently to the difference vector.
Accordingly, multiplying the single-coordinate generating functions over all
$i$ enumerates all possible ways of redistributing balls among the $q$ bins,
with the exponents of $s$ and $\tau$ recording the total numbers of additions
and removals, respectively.
This yields the global generating function
\[
   F_S(s,\tau)
   = \prod_{i=0}^{q-1} F_i(s,\tau).
\]

For integers $u,v\ge 0$, the coefficient
\[
   [s^u \tau^v]\,F_S(s,\tau)
\]
counts the number of difference vectors $\mathbf{z}$ satisfying
$\sum_i z_i^{+}=u$ and $\sum_i z_i^{-}=v$.
Moreover,
\[
   \sum_{i=0}^{q-1} z_i
   =
   \sum_{i=0}^{q-1} z_i^{+}-\sum_{i=0}^{q-1} z_i^{-}
   =
   u-v,
\]
so extracting the diagonal $u=v$ enforces the conservation constraint
$\sum_i z_i=0$ automatically.
Since feasible difference vectors must satisfy $u=v$, we obtain the following
exact formula.

\begin{lemma}[Generating-function formula for ball sizes]
\label{lem:ball-size-general}
For every $S\in\cS_{n,q}$ and every $r\ge 0$,
\[
   |B_r(S)|
   =
   \sum_{u=0}^{r}
      [s^u \tau^u]\,
      \prod_{i=0}^{q-1}
      \left(
         1 + \sum_{k\ge 1} s^k + \sum_{k=1}^{x_i} \tau^k
      \right).
\]
\end{lemma}

\medskip
\noindent\textbf{Illustration (binary case).}
Consider the binary case $q=2$.
Every multiset $S\in\cS_{n,2}$ is uniquely represented by its multiplicity
vector $\mathbf{x}_S=(x_0,x_1)$, where $x_0+x_1=n$.
The generating function of Lemma~\ref{lem:ball-size-general} then becomes
\[
   F_S(s,\tau)
   =
   \Bigl(1+\sum_{k\ge1}s^k+\sum_{k=1}^{x_0}\tau^k\Bigr)
   \Bigl(1+\sum_{k\ge1}s^k+\sum_{k=1}^{x_1}\tau^k\Bigr).
\]

A diagonal term $s^u \tau^u$ arises precisely by choosing a positive shift of
size $u$ in one coordinate of the difference vector and a negative shift of size
$u$ in the other coordinate.
Such a choice is feasible if and only if the negative shift does not exceed the
corresponding multiplicity, that is, if $u\le x_0$ or $u\le x_1$, respectively.
Consequently,
\[
   [s^u \tau^u]\,F_S(s,\tau)
   =
   \mathbf{1}\{u\le x_0\}
   +
   \mathbf{1}\{u\le x_1\},
   \qquad u\ge1,
\]
while $[s^0 \tau^0]\,F_S(s,\tau)=1$.

Summing over $u=0,\dots,r$ yields
\[
   |B_r(S)|
   =
   1+\sum_{u=1}^{r}
   \bigl(\mathbf{1}\{u\le x_0\}+\mathbf{1}\{u\le x_1\}\bigr).
\]
For an extreme multiset such as $S=\{0,0,\dots,0\}$ (whose multiplicity vector is
$(x_0,x_1)=(n,0)$), this expression reduces to
$|B_r(S)|=1+\min\{r,n\}$.
For more balanced multisets, both indicators may contribute over a larger range
of $u$, leading to larger ball sizes.

This computation illustrates how the generating-function formula captures the
dependence of deletion-ball sizes on the choice of the center.
In subsequent sections we formalize this intuition and prove that, in the binary
case, extreme multisets attain the minimal radius-$r$ ball size, while balanced
multisets attain the maximal one.

\subsection{Minimal balls occur at extreme centers}
\label{subsec:min-balls}

We now identify the centers that minimize the size of a radius-$r$ deletion ball
in $\cS_{n,q}$.

Throughout this subsection we assume $0 \le r < n$,
since for $r\ge n$ every radius-$r$ ball coincides with the entire space
$\cS_{n,q}$ and all centers yield the same ball size.

These two extremal statements are not merely heuristic consequences of the
geometry. Rather, they give a formal global characterization of the centers
that minimize and maximize deletion-ball size in $\cS_{n,q}$.

\medskip
\noindent
A shorter proof of the next lemma is possible. We nevertheless present it in the
completion-counting form below, because the same mechanism is needed again in
Lemma~\ref{lem:max-balanced}, where the argument is more delicate. In this way,
the proof of the minimal-ball statement also prepares the proof of the
maximal-ball statement.

\medskip
\noindent
\textbf{Shared local counting setup.}
Let $S\in\cS_{n,q}$ have multiplicity vector
\[
   \mathbf{x}=(x_0,\dots,x_{q-1}),
\]
and fix two coordinates $i,j$. Suppose that on this pair the center has values
$(\alpha,\beta)$, while all other coordinates are kept fixed. To count
multisets in $B_r(S)$, we use the canonical difference-vector decomposition
\begin{equation}
\label{eq:local-decomposition}
   \mathbf{z}=\mathbf{z}^+-\mathbf{z}^-\in\Delta_r(S),
\end{equation}
where
\[
   \mathbf{z}^+,\mathbf{z}^-\in\mathbb{Z}_{\ge0}^q,
   \qquad
   \supp(\mathbf{z}^+)\cap\supp(\mathbf{z}^-)=\varnothing,
   \qquad
   |\mathbf{z}^+|=|\mathbf{z}^-|\le r.
\]
For a fixed center $S$, each multiset $T\in B_r(S)$ corresponds to a unique
such vector $\mathbf{z}$.

We first fix the negative part outside the coordinates $i,j$. Thus fix
\begin{equation}
\label{eq:local-outside-pattern}
   \bar{\mathbf{u}}=(u_k)_{k\notin\{i,j\}}\in\mathbb{Z}_{\ge0}^{q-2},
\end{equation}
representing the coordinates of $\mathbf{z}^-$ outside the pair $(i,j)$, and
set
\begin{equation}
\label{eq:local-tlc}
   t\eqdef |\bar{\mathbf{u}}|\le r,
   \qquad
   \ell\eqdef |\supp(\bar{\mathbf{u}})|,
   \qquad
   c\eqdef r-t.
\end{equation}
Here $c$ is the remaining deletion budget on the pair $(i,j)$.

For $h\ge 1$, let $A_{\alpha,\beta}(h)$ be the number of ways to delete a total
of $h$ symbols from the pair $(\alpha,\beta)$ using exactly one of the two
coordinates, and let $B_{\alpha,\beta}(h)$ be the number of ways to delete a
total of $h$ symbols using both coordinates. Explicitly,
\begin{equation}
\label{eq:local-Aab}
   A_{\alpha,\beta}(h)=\mathbf{1}_{\{h\le \alpha\}}+\mathbf{1}_{\{h\le \beta\}},
\end{equation}
and
\begin{equation}
\label{eq:local-Bab}
   B_{\alpha,\beta}(h)
   =
   \bigl|\{u\in\{1,\dots,h-1\}: u\le \alpha,\ h-u\le \beta\}\bigr|.
\end{equation}

If the pair contributes total deletion weight $h$, then the total deletion
weight is $t+h$. When exactly one of the coordinates $i,j$ is used for
deletion, the positive part may be distributed over the remaining
$q-\ell-1$ available coordinates, giving
\begin{equation}
\label{eq:local-W1}
   W_1(h)
   \eqdef
   \binom{t+h+q-\ell-2}{q-\ell-2}
\end{equation}
possible completions. When both $i$ and $j$ are used for deletion, the positive
part may be distributed over the remaining $q-\ell-2$ coordinates, giving
\begin{equation}
\label{eq:local-W2}
   W_2(h)
   \eqdef
   \binom{t+h+q-\ell-3}{q-\ell-3},
\end{equation}
with the usual convention that
\[
   \binom{n}{k}=0 \qquad \text{whenever } k<0 \text{ or } k>n.
\]

For the fixed outside deletion pattern \eqref{eq:local-outside-pattern}, define
\begin{equation}
\label{eq:local-Cab}
   C_{\alpha,\beta}(t,\ell)
   \eqdef
   \binom{t+q-\ell-1}{q-\ell-1}
   +
   \sum_{h=1}^{c}\bigl(A_{\alpha,\beta}(h)W_1(h)+B_{\alpha,\beta}(h)W_2(h)\bigr).
\end{equation}
This is the number of admissible completions of the fixed outside negative part
$\bar{\mathbf{u}}$ to a full canonical difference vector in $\Delta_r(S)$. As
$\bar{\mathbf{u}}$ ranges over all feasible outside deletion patterns, these
completions exhaust $B_r(S)$ exactly once. Therefore,
\begin{equation}
\label{eq:local-ball-sum}
   |B_r(S)|
   =
   \sum_{\bar{\mathbf{u}}}
   C_{\alpha,\beta}\bigl(|\bar{\mathbf{u}}|,|\supp(\bar{\mathbf{u}})|\bigr),
\end{equation}
where the sum is over all feasible outside deletion patterns.

\begin{lemma}[Minimal balls at extreme centers]
\label{lem:extreme-min-ball}
Fix $n,q$ and an integer radius $0\le r<n$.
Among all centers $S\in\cS_{n,q}$, the size of the radius-$r$ ball $B_r(S)$
is minimized when $S$ is an extreme multiset,
\[
   S=\{a,a,\dots,a\}
   \qquad\text{for some } a\in\Sigma,
\]
that is, when the multiplicity vector is $(n,0,\dots,0)$ up to permutation.
For such a center,
\[
   |B_r(S)| = \binom{r+q-1}{q-1}.
\]
\end{lemma}

\begin{proof}
By symmetry of the alphabet, it suffices to consider the extreme center
\[
   S^{\star}=\{0^n\},
\]
whose multiplicity vector is $(n,0,\dots,0)$.

\smallskip
\noindent
We first compute the size of $B_r(S^{\star})$.
Any multiset $T\in B_r(S^{\star})$ is obtained by removing $u\le r$ balls from
the unique nonzero coordinate and then reinserting $u$ symbols arbitrarily from
the $q$-ary alphabet, allowing repetitions and in particular allowing symbols to
be placed back into the original coordinate.
Equivalently, starting from $(n,0,\dots,0)$, each such $T$ corresponds uniquely
to a multiset of size $u\le r$ over an alphabet of size $q$.
The number of such multisets equals the number of weak compositions of an
integer at most $r$ into $q$ parts, and therefore
\[
   |B_r(S^{\star})| = \binom{r+q-1}{q-1}.
\]

\smallskip
\noindent
We now show that an \emph{unbalancing move} cannot increase the ball size.
Let $S\in\cS_{n,q}$ have multiplicity vector $\mathbf{x}=(x_0,\dots,x_{q-1})$.
If $S$ is not extreme, then at least two coordinates are positive. Choose
indices $i,j$ such that
\[
   a\eqdef x_i \ge x_j \eqdef b \ge 1,
\]
and let $S'$ be obtained by replacing the pair $(a,b)$ with $(a+1,b-1)$ while
leaving all other coordinates unchanged. It suffices to prove that
\[
   |B_r(S')|\le |B_r(S)|.
\]

By the shared setup above, and in particular by \eqref{eq:local-ball-sum},
for each fixed outside deletion pattern it is enough to compare
$C_{a+1,b-1}(t,\ell)$ and $C_{a,b}(t,\ell)$. Since the first term in
\eqref{eq:local-Cab} corresponds to $h=0$, it is independent of the pair.
Moreover,
\[
   A_{a+1,b-1}(h)-A_{a,b}(h)
   =
   \mathbf{1}_{\{h=a+1\}}-\mathbf{1}_{\{h=b\}},
\]
and
\[
   B_{a+1,b-1}(h)-B_{a,b}(h)
   =
   \begin{cases}
      -1, & b+1\le h\le a+1,\\
      0, & \text{otherwise}.
   \end{cases}
\]
Hence
\[
\begin{aligned}
   C_{a+1,b-1}(t,\ell)-C_{a,b}(t,\ell)
   &=
   \mathbf{1}_{\{c\ge a+1\}}W_1(a+1)
   -
   \mathbf{1}_{\{c\ge b\}}W_1(b) \\
   &\qquad -
   \sum_{h=b+1}^{\min(c,a+1)} W_2(h).
\end{aligned}
\]
By Pascal's identity,
\[
   W_1(h)-W_1(h-1)=W_2(h) \qquad (h\ge1),
\]
so
\[
   W_1(a+1)-W_1(b)=\sum_{h=b+1}^{a+1} W_2(h).
\]
If $c<b$, then the difference is $0$. If $b\le c<a+1$, then
\[
\begin{aligned}
   C_{a+1,b-1}(t,\ell)-C_{a,b}(t,\ell)
   &= -W_1(b)-\sum_{h=b+1}^{c}W_2(h) \\
   &\le 0.
\end{aligned}
\]
If $c\ge a+1$, then
\[
\begin{aligned}
   C_{a+1,b-1}(t,\ell)-C_{a,b}(t,\ell)
   &= W_1(a+1)-W_1(b)-\sum_{h=b+1}^{a+1}W_2(h) \\
   &=0.
\end{aligned}
\]
Thus in all cases
\[
   C_{a+1,b-1}(t,\ell)\le C_{a,b}(t,\ell).
\]
Summing over all outside deletion patterns in \eqref{eq:local-ball-sum} gives
\[
   |B_r(S')|\le |B_r(S)|.
\]

Starting from any non-extreme multiplicity vector, repeated unbalancing moves
eventually produce an extreme vector, and at each step the radius-$r$ ball
size does not increase. Therefore every extreme multiset minimizes
$|B_r(S)|$, and its value is $\binom{r+q-1}{q-1}$, as claimed.
\end{proof}

\subsection{Maximal balls and the set $S_{q-1}(r,r)$}
\label{subsec:max-balls}

We now relate metric balls in $\cS_{n,q}$ to the combinatorial set
$S_{q-1}(r,r)$ introduced by Kova\v{c}evi\'c and Tan~\cite{KovacevicTan2018},
and clarify when this set exactly describes the collection of feasible
difference vectors of a radius-$r$ ball.

Let $S\in\cS_{n,q}$ be a multiset with multiplicity vector
$\mathbf{x}_S=(x_0,\dots,x_{q-1})\in\mathbb{Z}_{\ge 0}^q$.
Recall from Section~\ref{subsec:ball-general-formula} that the difference
vectors corresponding to multisets at deletion distance at most $r$ from $S$
form the set
\[
   \Delta_r(S)
   =
   \{\mathbf{x}_T-\mathbf{x}_S : T\in\cS_{n,q},\ d(S,T)\le r\}.
\]

\begin{lemma}[Containment in the ideal difference set]
\label{lem:Sr-containment}
For every center $S\in\cS_{n,q}$ and every radius $r\ge 0$,
\[
   \Delta_r(S)\subseteq S_{q-1}(r,r).
\]
Consequently,
\[
   |B_r(S)| = |\Delta_r(S)| \le |S_{q-1}(r,r)|.
\]
\end{lemma}

\begin{proof}
Let $\mathbf{z}=\mathbf{x}_T-\mathbf{x}_S$ for some $T\in B_r(S)$. By the decomposition $\mathbf{z}=\mathbf{z}^{+}-\mathbf{z}^{-}$ introduced earlier, we have
\[
   \sum_{i=0}^{q-1} z_i^{+}
   =
   \sum_{i=0}^{q-1} z_i^{-}
   =
   d(S,T).
\]
Since $T\in B_r(S)$, it follows that
\[
   \sum_{i=0}^{q-1} z_i^{+} \le r
   \qquad\text{and}\qquad
   \sum_{i=0}^{q-1} z_i^{-} \le r.
\]
Moreover, because both $S$ and $T$ have total multiplicity $n$,
\[
   \sum_{i=0}^{q-1} z_i = 0.
\]
Hence $\mathbf{z}\in S_{q-1}(r,r)$, proving the claimed containment.
\end{proof}

\begin{remark}
Every difference vector $\mathbf{z}\in\Delta_r(S)$ satisfies the zero-sum constraint because both the center and the received multiset have the same total size. Thus $S_{q-1}(r,r)$ should be viewed as an \emph{ideal ambient container} for feasible difference vectors. Equality need not hold in general: near the boundary of the simplex, some vectors in $S_{q-1}(r,r)$ violate the feasibility condition $\mathbf{z}^{-}\le \mathbf{x}_S$ and therefore do not correspond to valid points of $B_r(S)$. This boundary truncation will be important later in Section~\ref{sec:vol-bounds}.
\end{remark}

The cardinality of $S_{q-1}(r^{+},r^{-})$ is given explicitly in
\cite[Lemma~1]{KovacevicTan2018} by
\[
   |S_{q-1}(r^{+},r^{-})|
   =
   \sum_{j=0}^{r^{+}}
      \binom{q-1}{j}
      \binom{r^{+}}{j}
      \binom{ r^{-}+q-1-j}{q-1-j}.
\]
In the symmetric case $r^{+}=r^{-}=r$, this agrees with the one-variable
specialization in Corollary~\ref{cor:Sr-explicit} below.
\begin{corollary}[Interior centers attain the ideal upper bound]
\label{cor:interior-ball}
If the center $S$ satisfies
\[
   x_i \ge r \qquad \text{for all } i\in\{0,\dots,q-1\},
\]
then
\[
   \Delta_r(S)=S_{q-1}(r,r)
   \qquad\text{and}\qquad
   |B_r(S)|=|S_{q-1}(r,r)|.
\]
In particular, such centers exist whenever $n\ge qr$.
\end{corollary}

\begin{proof}
If $x_i\ge r$ for every $i$, then for every $\mathbf{z}\in S_{q-1}(r,r)$ one has
$\sum_i z_i^-\le r$ and therefore $z_i^-\le r\le x_i$ coordinatewise. Thus the
feasibility constraint $\mathbf{z}^-\le \mathbf{x}_S$ is automatic, so every
vector in $S_{q-1}(r,r)$ is realized by a valid multiset in $B_r(S)$.
Combined with Lemma~\ref{lem:Sr-containment}, this yields the claim.
\end{proof}

If some coordinate $x_i<r$, then certain vectors in $S_{q-1}(r,r)$ violate
the feasibility condition
$\mathbf{z}^{-}\le\mathbf{x}_S$.
In this regime the containment of Lemma~\ref{lem:Sr-containment} is strict, and
\[
   |B_r(S)| < |S_{q-1}(r,r)|.
\]
Thus, while $S_{q-1}(r,r)$ always provides a universal combinatorial upper bound
on the size of a radius-$r$ ball, this bound is attained only when the center is
sufficiently far from the boundary of $\cS_{n,q}$.

\paragraph{Example (ternary case, $n=6$).}
Consider $\cS_{6,3}$, whose cardinality is
\[
   |\cS_{6,3}|=\binom{8}{2}=28.
\]
We compare radius-$r$ ball sizes for three representative centers:
an extreme center $(6,0,0)$, a moderately unbalanced center $(3,2,1)$,
and the balanced center $(2,2,2)$.
A direct enumeration yields:
\[
\begin{array}{c|ccc}
r
& |B_r(6,0,0)|
& |B_r(3,2,1)|
& |B_r(2,2,2)| \\ \hline
0 & 1  & 1  & 1 \\
1 & 3  & 7  & 7 \\
2 & 6  & 16 & 19 \\
3 & 10 & 24 & 25 \\
4 & 15 & 27 & 28
\end{array}
\]

The extreme center attains the minimal ball size for all $r<n$.
For $r=1$ the balanced and unbalanced interior centers tie.
Starting from $r=2$, the balanced center $(2,2,2)$ strictly maximizes the
ball size and reaches the entire space earlier.
This example illustrates the effect of boundary truncation and is consistent
with the general maximality statement proved below.

\begin{definition}[Most balanced multiset]
A multiset $S\in\cS_{n,q}$ with multiplicity vector
$\mathbf{x}_S=(x_0,\dots,x_{q-1})$ is called \emph{most balanced} if
\[
   \max_{0\le i\le q-1} x_i \;-\; \min_{0\le i\le q-1} x_i \;\le\; 1.
\]
Equivalently, letting
\[
   a \eqdef \left\lfloor \frac{n}{q} \right\rfloor,
   \qquad
   b \eqdef \left\lceil \frac{n}{q} \right\rceil = a+1,
\]
there exists an integer $k\in\{0,1,\dots,q-1\}$ such that
\[
   \{x_0,\dots,x_{q-1}\}
   =
   \{\underbrace{b,\dots,b}_{k\ \text{times}},
     \underbrace{a,\dots,a}_{q-k\ \text{times}}\},
\]
where necessarily $k = n - qa$.
\end{definition}

\begin{lemma}[Most balanced centers maximize deletion balls]
\label{lem:max-balanced}
For every $n,q,r$, a most balanced multiset $S\in\cS_{n,q}$ maximizes
$|B_r(S)|$ among all centers in $\cS_{n,q}$.
\end{lemma}

\begin{proof}
Let $S\in\cS_{n,q}$ have multiplicity vector
$\mathbf{x}=(x_0,\dots,x_{q-1})$. If $\mathbf{x}$ is not most balanced, then
there exist indices $i,j$ such that
\[
   a\eqdef x_i \ge x_j+2 \eqdef b+2.
\]
Let $S'$ be the center obtained by replacing the pair $(a,b)$ with
$(a-1,b+1)$ and leaving all other coordinates unchanged. It suffices to show
that
\[
   |B_r(S')|\ge |B_r(S)|.
\]

Using the shared setup before Lemma~\ref{lem:extreme-min-ball}, and in
particular \eqref{eq:local-ball-sum}, we may compare the two centers through
$C_{a,b}(t,\ell)$ and $C_{a-1,b+1}(t,\ell)$ for the same outside deletion
pattern \eqref{eq:local-outside-pattern}. Since the first term in
\eqref{eq:local-Cab} is independent of the pair, it remains only to compare the
contributions for $h\ge1$. A direct check gives
\[
   A_{a-1,b+1}(h)-A_{a,b}(h)
   =
   \mathbf{1}_{\{h=b+1\}}-\mathbf{1}_{\{h=a\}},
\]
and
\[
   B_{a-1,b+1}(h)-B_{a,b}(h)
   =
   \begin{cases}
      1, & b+2\le h\le a,\\
      0, & \text{otherwise}.
   \end{cases}
\]
Therefore,
\[
\begin{aligned}
   C_{a-1,b+1}(t,\ell)-C_{a,b}(t,\ell)
   &=
   \mathbf{1}_{\{c\ge b+1\}}W_1(b+1)
   -
   \mathbf{1}_{\{c\ge a\}}W_1(a) \\
   &\qquad +
   \sum_{h=b+2}^{\min(c,a)} W_2(h).
\end{aligned}
\]
Again using Pascal's identity,
\[
   W_1(h)-W_1(h-1)=W_2(h) \qquad (h\ge1),
\]
we obtain
\[
   W_1(a)-W_1(b+1)=\sum_{h=b+2}^{a} W_2(h).
\]
If $c<a$, then
\[
   C_{a-1,b+1}(t,\ell)-C_{a,b}(t,\ell)
   =
   \mathbf{1}_{\{c\ge b+1\}}W_1(b+1)+\sum_{h=b+2}^{c} W_2(h)
   \ge 0.
\]
If $c\ge a$, then
\[
\begin{aligned}
   C_{a-1,b+1}(t,\ell)-C_{a,b}(t,\ell)
   &=
   W_1(b+1)-W_1(a)+\sum_{h=b+2}^{a} W_2(h) \\
   &= 0.
\end{aligned}
\]
Hence for every outside deletion pattern,
\[
   C_{a-1,b+1}(t,\ell)\ge C_{a,b}(t,\ell).
\]
Summing over all outside deletion patterns in \eqref{eq:local-ball-sum} yields
\[
   |B_r(S')|\ge |B_r(S)|.
\]

Starting from an arbitrary multiplicity vector of total sum $n$, repeated
balancing moves eventually produce a most balanced vector, and at each step the
radius-$r$ ball size does not decrease. Consequently, a most balanced multiset
maximizes $|B_r(S)|$ among all centers in $\cS_{n,q}$.
\end{proof}

\begin{remark}[Interior case and boundary truncation]
The lemma identifies the correct maximizers globally. In the fully interior case
$x_i\ge r$ for all $i$, the maximal value equals the universal upper bound
$|S_{q-1}(r,r)|$. When $n<qr$, no center is fully interior and the maximal ball
size is smaller because of boundary truncation; nevertheless, the lemma shows
that the maximizing centers are still the most balanced ones.
\end{remark}

\subsection{Average ball size via generating functions}
\label{subsec:avg-ball-general}

We now turn to the \emph{average} radius-$r$ deletion ball size in $\cS_{n,q}$.
The key point is that the same generating-function framework used for individual
centers also yields an exact formula for the number of ordered pairs of multisets
at a prescribed distance, and hence for the average ball size itself.

Recall from Section~\ref{sec:pre} that
\[
   \overline{B_r}(n,q)
   =
   \frac{1}{|\cS_{n,q}|}
   \sum_{S\in\cS_{n,q}} |B_r(S)|
\]
and that $A_{n,q}(m)$ denotes the number of ordered pairs of multisets in
$\cS_{n,q}$ at deletion distance exactly $m$. Then
\[
   \sum_{S\in\cS_{n,q}} |B_r(S)|
   = \sum_{m=0}^{r} A_{n,q}(m).
\]
It is therefore natural to introduce the bivariate generating function
\[
   Q_q(u,z)
   \;\eqdef\;
   \sum_{n\ge 0}\sum_{m\ge 0} A_{n,q}(m)\,u^n z^m .
\]

Recall from Section~\ref{subsec:ball-general-formula} that for a fixed center
$S\in\cS_{n,q}$ with multiplicity vector
$\mathbf{x}_S=(x_0,\dots,x_{q-1})$, the radius-$r$ ball size admits the
representation
\[
   |B_r(S)|
   =
   \sum_{m=0}^{r}
      [s^m \tau^m]\,
      F_S(s,\tau),
\]
where
\[
   F_S(s,\tau)
   =
   \prod_{i=0}^{q-1}
   F_{x_i}(s,\tau),
   \qquad
   F_x(s,\tau)
   =
   1+\sum_{k\ge1}s^k+\sum_{k=1}^{x}\tau^k .
\]

As before, define
\[
   G(u;s,\tau)
   \;\eqdef\;
   \sum_{x\ge0} u^x\,F_x(s,\tau).
\]
Using
\[
   \sum_{k\ge1}s^k=\frac{s}{1-s},
   \qquad
   \sum_{k=1}^{x}\tau^k=\frac{\tau(1-\tau^{x})}{1-\tau},
\]
one finds
\[
\begin{aligned}
G(u;s,\tau)
&=
\sum_{x\ge0}u^x
\left(
1+\frac{s}{1-s}+\frac{\tau(1-\tau^x)}{1-\tau}
\right) \\[0.5ex]
&=
\left(\frac{1}{1-s}+\frac{\tau}{1-\tau}\right)\frac{1}{1-u}
-
\frac{\tau}{(1-\tau)(1-u\tau)}.
\end{aligned}
\]
Since
\[
   [u^n]\bigl(G(u;s,\tau)\bigr)^q
   =
   \sum_{S\in\cS_{n,q}} F_S(s,\tau),
\]
we obtain
\[
   Q_q(u,z)
   =
   \sum_{m\ge0} z^m [s^m\tau^m]\bigl(G(u;s,\tau)\bigr)^q.
\]

The diagonal extraction in $(s,\tau)$ can be reduced to a one-variable constant
term. Indeed, if
\[
   H(s,\tau)=\sum_{a,b\ge0} h_{a,b}s^a\tau^b,
\]
then
\[
   [s^0]\,H\!\left(s,\frac{z}{s}\right)
   =
   \sum_{a,b\ge0} h_{a,b} z^b [s^{a-b}]
   =
   \sum_{m\ge0} h_{m,m} z^m
   =
   \sum_{m\ge0} z^m [s^m\tau^m] H(s,\tau).
\]
Applying this identity with $H(s,\tau)=\bigl(G(u;s,\tau)\bigr)^q$ yields
\[
   Q_q(u,z)
   =
   [s^0]\,
   \bigl(G(u;s,z/s)\bigr)^q.
\]

We now simplify the single-variable expression. Substituting $\tau=z/s$ into
the formula for $G$ and simplifying gives
\[
   G\!\left(u;s,\frac{z}{s}\right)
   =
   \frac{s(1-uz)}{(1-u)(1-s)(s-uz)}.
\]
Therefore,
\[
   Q_q(u,z)
   =
   [s^0]\,
   \left(
      \frac{s(1-uz)}{(1-u)(1-s)(s-uz)}
   \right)^q .
\]

The constant term can now be extracted explicitly.

\begin{theorem}[Exact pair enumerator and average balls]
\label{thm:avg-ball-general}
Let
\[
   P_{q-1}(x)
   \;\eqdef\;
   \sum_{j=0}^{q-1}\binom{q-1}{j}^2 x^j,
\]
and define coefficients $c_{q,m}$ by
\[
   \frac{P_{q-1}(x)}{(1-x)^{q-1}}
   =
   \sum_{m\ge0} c_{q,m}x^m.
\]
Equivalently,
\[
   c_{q,m}
   =
   \sum_{j=0}^{\min(m,q-1)}
   \binom{q-1}{j}^2
   \binom{m-j+q-2}{q-2}.
\]
Then
\[
   Q_q(u,z)
   =
   \frac{P_{q-1}(uz)}{(1-u)^q(1-uz)^{q-1}}
   =
   \frac{1}{(1-u)^q}\sum_{m\ge0} c_{q,m}(uz)^m.
\]
Consequently,
\[
   A_{n,q}(m)
   =
   c_{q,m}\binom{n-m+q-1}{q-1},
   \qquad 0\le m\le n,
\]
and
\[
   \overline{B_r}(n,q)
   =
   \frac{1}{\binom{n+q-1}{q-1}}
   \sum_{m=0}^{r}
   c_{q,m}\binom{n-m+q-1}{q-1}.
\]
\end{theorem}

\begin{proof}
From the constant-term representation above,
\[
   Q_q(u,z)
   =
   [s^0]\,
   \left(
      \frac{s(1-uz)}{(1-u)(1-s)(s-uz)}
   \right)^q
   =
   \frac{(1-uz)^q}{(1-u)^q}
   [s^0]\,
   \frac{s^q}{(1-s)^q(s-uz)^q}.
\]
Now write
\[
   \frac{1}{(s-uz)^q}
   =
   s^{-q}\left(1-\frac{uz}{s}\right)^{-q}
   =
   s^{-q}\sum_{m\ge0}\binom{q+m-1}{m}(uz)^m s^{-m},
\]
and also
\[
   \frac{1}{(1-s)^q}
   =
   \sum_{\ell\ge0}\binom{q+\ell-1}{\ell}s^\ell.
\]
Multiplying these two series, the coefficient of $s^0$ is obtained precisely
when $\ell=m$. Hence
\[
   Q_q(u,z)
   =
   \frac{(1-uz)^q}{(1-u)^q}
   \sum_{m\ge0}\binom{q+m-1}{m}^2(uz)^m.
\]

Using the classical identity
\[
   \sum_{m\ge0}\binom{q+m-1}{m}^2 x^m
   =
   \frac{P_{q-1}(x)}{(1-x)^{2q-1}},
\]
we obtain
\[
   Q_q(u,z)
   =
   \frac{P_{q-1}(uz)}{(1-u)^q(1-uz)^{q-1}}.
\]

Next, expand
\[
   \frac{P_{q-1}(uz)}{(1-uz)^{q-1}}
   =
   \sum_{m\ge0} c_{q,m}(uz)^m.
\]
Then
\[
   Q_q(u,z)
   =
   \frac{1}{(1-u)^q}\sum_{m\ge0} c_{q,m}(uz)^m,
\]
so extracting the coefficient of $u^n z^m$ gives
\[
   A_{n,q}(m)
   =
   c_{q,m}[u^{n-m}](1-u)^{-q}
   =
   c_{q,m}\binom{n-m+q-1}{q-1}.
\]
Finally,
\[
   \sum_{S\in\cS_{n,q}}|B_r(S)|
   =
   \sum_{m=0}^{r}A_{n,q}(m),
\]
and division by $|\cS_{n,q}|=\binom{n+q-1}{q-1}$ yields the formula for
$\overline{B_r}(n,q)$.

The explicit expression for $c_{q,m}$ follows by expanding
\(P_{q-1}(x)(1-x)^{-(q-1)}\) and collecting the coefficient of \(x^m\).
\end{proof}

The preceding theorem shows that the average-ball problem is governed by a
one-variable generating function rather than by a genuinely bivariate diagonal
at the final stage. In particular, it gives an exact finite-$n$ formula for the
distance distribution and for the average ball size.

\begin{corollary}[Exact size of the ideal difference set]
\label{cor:Sr-explicit}
For every $q\ge 2$ and $r\ge 0$,
\[
   |S_{q-1}(r,r)|
   =
   \sum_{m=0}^{r} c_{q,m}
   =
   [x^r]\frac{P_{q-1}(x)}{(1-x)^q}
   =
   \sum_{j=0}^{\min(r,q-1)}
      \binom{q-1}{j}^2
      \binom{r-j+q-1}{q-1}.
\]
In particular, this recovers the exact formula for $|S_{q-1}(r,r)|$ already
obtained by Kova\v{c}evi\'c and Tan~\cite[Lemma~1]{KovacevicTan2018}.
\end{corollary}

\begin{proof}
Summing the identity
\[
   \frac{P_{q-1}(x)}{(1-x)^{q-1}}
   =
   \sum_{m\ge0} c_{q,m}x^m
\]
from $m=0$ to $m=r$ is equivalent to multiplying by $(1-x)^{-1}$ and extracting
the coefficient of $x^r$, which gives
\[
   |S_{q-1}(r,r)|
   =
   [x^r]\frac{P_{q-1}(x)}{(1-x)^q}.
\]
Expanding $(1-x)^{-q}$ and collecting the coefficient of $x^r$ yields the stated
closed form. The agreement with~\cite[Lemma~1]{KovacevicTan2018} is immediate.
\end{proof}

\begin{corollary}[Immediate specializations]
\label{cor:avg-ball-specializations}
The exact formulas above yield the following explicit cases.

\begin{enumerate}
\item For $q=2$,
\[
   A_{n,2}(0)=n+1,
   \qquad
   A_{n,2}(m)=2(n-m+1)\quad (1\le m\le n),
\]
\[
   \overline{B_r}(n,2)
   =
   (2r+1)-\frac{r(r+1)}{n+1},
\]
and
\[
   |S_1(r,r)|=2r+1.
\]

\item For $q=3$,
\[
   A_{n,3}(0)=\binom{n+2}{2},
   \qquad
   A_{n,3}(m)=6m\binom{n-m+2}{2}\quad (1\le m\le n),
\]
\[
   \overline{B_r}(n,3)
   =
   \frac{1}{\binom{n+2}{2}}
   \left(
      \binom{n+2}{2}
      +
      6\sum_{m=1}^{r} m\binom{n-m+2}{2}
   \right),
\]
and
\[
   |S_2(r,r)|=3r^2+3r+1.
\]
\end{enumerate}
\end{corollary}

\begin{proof}
For $q=2$, one has $P_1(x)=1+x$, so
\[
   \frac{P_1(x)}{1-x}
   =1+2\sum_{m\ge1}x^m.
\]
Hence $c_{2,0}=1$ and $c_{2,m}=2$ for $m\ge1$, which gives the stated formulas.
For $q=3$, one has $P_2(x)=1+4x+x^2$, and
\[
   \frac{P_2(x)}{(1-x)^2}
   =1+6\sum_{m\ge1} m x^m.
\]
Hence $c_{3,0}=1$ and $c_{3,m}=6m$ for $m\ge1$, and the claims follow from
Theorem~\ref{thm:avg-ball-general} and Corollary~\ref{cor:Sr-explicit}.
\end{proof}

\begin{remark}
For $q=4$, the same computation gives $c_{4,0}=1$ and
$c_{4,m}=10m^2+2$ for $m\ge1$, and therefore
\[
   |S_3(r,r)|
   =
   \frac{10r^3+15r^2+11r+3}{3}.
\]
\end{remark}

We now record the asymptotic consequence needed later for the
Gilbert--Varshamov bound.

\begin{lemma}[Average ball size is asymptotically ideal]
\label{lem:avg-ball-ideal}
Fix $q\ge2$ and $r\ge0$, and let $n\to\infty$.
Then
\[
   \overline{B_r}(n,q)
   =
   |S_{q-1}(r,r)| - O\!\left(\frac{1}{n}\right).
\]
More precisely,
\[
   \overline{B_r}(n,q)
   =
   \sum_{m=0}^{r} c_{q,m}
   \frac{\binom{n-m+q-1}{q-1}}{\binom{n+q-1}{q-1}}
   =
   \sum_{m=0}^{r} c_{q,m} + O\!\left(\frac{1}{n}\right),
\]
and
\[
   \sum_{m=0}^{r} c_{q,m}
   =
   |S_{q-1}(r,r)|.
\]
\end{lemma}

\begin{proof}
By Theorem~\ref{thm:avg-ball-general},
\[
   \overline{B_r}(n,q)
   =
   \sum_{m=0}^{r} c_{q,m} R_{n,q}(m),
   \qquad
   R_{n,q}(m)
   \;\eqdef\;
   \frac{\binom{n-m+q-1}{q-1}}{\binom{n+q-1}{q-1}}.
\]
For fixed $m$ and $q$,
\[
   R_{n,q}(m)
   =
   \prod_{j=0}^{q-2}\frac{n-m+1+j}{n+1+j}
   =
   \prod_{j=0}^{q-2}\left(1-\frac{m}{n+1+j}\right)
   =
   1+O\!\left(\frac{1}{n}\right),
\]
uniformly for $0\le m\le r$.
Since the sum over $m=0,\dots,r$ is finite, this gives
\[
   \overline{B_r}(n,q)
   =
   \sum_{m=0}^{r} c_{q,m} + O\!\left(\frac{1}{n}\right).
\]

It remains to identify the constant term \(\sum_{m=0}^{r} c_{q,m}\).
If \(S\in\cS_{n,q}\) satisfies \(\min_i S(i)\ge r\), then the boundary constraint
is inactive for every difference vector of radius at most \(r\), and
Lemma~\ref{lem:ball-size-general} counts exactly all vectors in the ideal
difference set \(S_{q-1}(r,r)\).
Therefore the number of admissible vectors at exact distance \(m\) is \(c_{q,m}\),
and summing over \(m\le r\) gives
\[
   \sum_{m=0}^{r} c_{q,m}
   =
   |S_{q-1}(r,r)|.
\]
This proves the claim.
\end{proof}

\subsection{Ball sizes for $q=2$ and $q=3$}
\label{subsec:balls-q2-q3}

In this subsection we collect explicit formulas for minimal, maximal, and
average radius-$r$ deletion balls for the binary and ternary alphabets.
The average-ball identities are consistent with the immediate specializations in
Corollary~\ref{cor:avg-ball-specializations}; here we place them alongside the
corresponding minimal- and maximal-ball formulas.

\subsubsection{Binary alphabet ($q=2$)}

We start with the following lemma.
\begin{lemma}
\label{lem:binary-distance}
Let $q=2$ and $S,T\in\cS_{n,2}$.
Then the deletion distance depends only on the weights and satisfies
\[
   d(S,T)=|w(S)-w(T)|.
\]
\end{lemma}

\begin{proof}
Write the multiplicity vectors as
\[
   \mathbf{x}_S=(n-w(S),w(S)), \qquad
   \mathbf{x}_T=(n-w(T),w(T)).
\]
Then the difference vector is
\[
   \mathbf{z}=\mathbf{x}_T-\mathbf{x}_S=(w(S)-w(T),\,w(T)-w(S)).
\]
Hence
\[
   d(S,T)=\sum_i z_i^+=|w(S)-w(T)|,
\]
as claimed.
\end{proof}

\begin{lemma}[Binary ball sizes]
\label{lem:q2-balls-summary}
Fix $n\ge 1$ and $0\le r\le n$.
For a center $S\in\cS_{n,2}$ with $w=w_S$, the radius-$r$ ball size is
\[
   |B_r(S)|
   = \min\{n,\;w+r\}-\max\{0,\;w-r\}+1.
\]
Consequently,
\[
   \min_{S\in\cS_{n,2}} |B_r(S)|=r+1,
\]
attained at the extreme centers $w=0$ and $w=n$.
Moreover,
\[
   \max_{S\in\cS_{n,2}} |B_r(S)|
   =
   \begin{cases}
      2r+1, & r\le \lfloor n/2\rfloor,\\[1mm]
      n+1,  & r\ge \lceil n/2\rceil,
   \end{cases}
\]
and the maximum is attained by every center satisfying $r\le w\le n-r$ when
this interval is nonempty, and otherwise by the most balanced centers.

Finally, the average ball size equals
\[
   \overline{B_r}(n,2)
   = (2r+1)-\frac{r(r+1)}{n+1}.
\]
\end{lemma}

\begin{proof}
Since every feasible difference vector has the form $(t,-t)$, we have
$d(S,T)=|w_T-w_S|$.
Thus
\[
   B_r(S)=\{\,T\in\cS_{n,2}:\ |w_T-w|\le r\,\}.
\]
As there is exactly one multiset of each weight $0,\dots,n$, this gives
\[
   |B_r(S)|
   =\bigl|\{\,j\in\{0,\dots,n\}:\ |j-w|\le r\,\}\bigr|
   =\min\{n,w+r\}-\max\{0,w-r\}+1.
\]

The minimal value $r+1$ is attained at $w=0,n$, in agreement with
Lemma~\ref{lem:extreme-min-ball}.
If $r\le \lfloor n/2\rfloor$, then centers with $r\le w\le n-r$ exist and yield
$|B_r(S)|=2r+1$.
If $r\ge \lceil n/2\rceil$, every radius-$r$ ball equals the entire space,
hence $|B_r(S)|=n+1$.

For the average, counting ordered pairs $(w_S,w_T)\in\{0,\dots,n\}^2$ with
$|w_S-w_T|\le r$ yields
\[
   \sum_{S\in\cS_{n,2}} |B_r(S)|
   =(n+1)(2r+1)-r(r+1),
\]
and division by $|\cS_{n,2}|=n+1$ gives the stated formula.
\end{proof}

\subsubsection{Ternary alphabet ($q=3$)}

Let $q=3$ with alphabet $\Sigma=\{0,1,2\}$.
Recall that a multiset $S\in\cS_{n,3}$ is represented by its multiplicities
\[
   \mathbf{x}_S=(x_0,x_1,x_2)\in\Z_{\ge 0}^3,
   \qquad
   x_0+x_1+x_2=n.
\]

\begin{lemma}[Ternary ball sizes]
\label{lem:q3-balls-summary}
Fix $n\ge 1$ and $r\ge 0$.

The minimal radius-$r$ ball size is attained at extreme centers and equals
\[
   \min_{S\in\cS_{n,3}} |B_r(S)|
   = \binom{r+2}{2}
   = \frac{(r+1)(r+2)}{2}.
\]

If $n\ge 3r$ and $\min_i x_i\ge r$, then
\[
   |B_r(S)|
   = 3r^2+3r+1.
\]
In particular, centers whose multiplicities are as equal as possible attain
the maximal ball size in this regime.

Finally, for fixed $r$ and $n\to\infty$,
\[
   \overline{B_r}(n,3)
   = (3r^2+3r+1) - O\!\left(\frac{1}{n}\right).
\]
\end{lemma}

\begin{proof}
The minimal-ball formula is the specialization of
Lemma~\ref{lem:extreme-min-ball} to $q=3$.

Assume $n\ge 3r$ and $\min_i x_i\ge r$.
Recall from Lemma~\ref{lem:ball-size-general} that
\[
   B_r(S)
   =
   \{\mathbf{z}\in\Z^3 :
     \sum_i z_i=0,\;
     \sum_i z_i^+ \le r,\;
     \sum_i z_i^- \le r,\;
     \mathbf{z}^- \le \mathbf{x}_S
   \}.
\]
When $\min_i x_i\ge r$, the feasibility constraint
$\mathbf{z}^-\le \mathbf{x}_S$ is inactive for all vectors satisfying
$\sum_i z_i^-\le r$.
Hence the admissible difference vectors are exactly those satisfying
\[
   \sum_i z_i=0
   \qquad\text{and}\qquad
   \sum_i z_i^+=\sum_i z_i^-\le r.
\]

For $q=3$, these vectors form a discrete hexagonal region in the $A_2$ lattice.
More precisely, for $t\ge1$, the number of vectors $\mathbf{z}\in\Z^3$ with
$\sum_i z_i=0$ and $\sum_i z_i^+=\sum_i z_i^-=t$ equals $6t$.
Indeed, one chooses the unique negative coordinate and distributes $t$
between the two positive coordinates, yielding $3\cdot 2\cdot t=6t$ possibilities.

Summing over $t=0,\dots,r$ yields
\[
   |B_r(S)|
   = 1+\sum_{t=1}^r 6t
   = 3r^2+3r+1.
\]

Balanced centers satisfy $\min_i x_i\ge r$ whenever $n\ge 3r$, and therefore
attain this maximal value.

The asymptotic average-ball estimate also follows from
Corollary~\ref{cor:avg-ball-specializations} by observing that
\[
   \frac{1}{\binom{n+2}{2}}
   \left(
      \binom{n+2}{2}
      +
      6\sum_{m=1}^{r} m\binom{n-m+2}{2}
   \right)
   =
   3r^2+3r+1-O\!\left(\frac{1}{n}\right),
\]
which is consistent with Lemma~\ref{lem:avg-ball-ideal}.
\end{proof}

Finally, we note that while the binary case admits a fully explicit average
formula, the generating-function framework of
Theorem~\ref{thm:avg-ball-general} provides a unified approach to average ball
sizes for all $q$, and yields the asymptotic description above for $q=3$.


\section{Volume-Based Bounds for Multiset Deletion Codes}
\label{sec:vol-bounds}

In this section we derive general upper and lower bounds on the maximal
cardinality of multiset deletion-correcting codes.
Throughout, we consider the multiset space $\cS_{n,q}$ equipped with the
deletion distance
\[
   d(S,T)=n-|S\cap T|.
\]

All bounds in this section are expressed in terms of deletion balls, whose
structure and size were analyzed in Section~\ref{sec:balls}.
A key feature of the multiset model is that ball sizes depend strongly on
the choice of the center.
As a result, different bounding techniques rely on different extremal
notions of ball size:
\begin{itemize}
\item \emph{minimal} balls govern sphere-packing upper bounds;
\item \emph{maximal} balls motivate code--anticode style arguments in the
      difference-vector space, which are valid only in specific geometric regimes;
\item \emph{average} balls govern Gilbert--Varshamov-type lower bounds.
\end{itemize}

We treat these bounds in turn and clarify their scope and limitations.

\subsection{Sphere-packing upper bound}
\label{subsec:sp-bound}

Let $\cC\subseteq\cS_{n,q}$ be a multiset code with minimum distance
$d(\cC)=d$.
Set
\[
   r=\Bigl\lfloor\frac{d-1}{2}\Bigr\rfloor .
\]
By definition of the deletion distance, the radius-$r$ balls
$\{B_r(S): S\in\cC\}$ are pairwise disjoint.
Since all balls lie in the ambient space $\cS_{n,q}$, whose cardinality is
\[
   |\cS_{n,q}|=\binom{n+q-1}{q-1},
\]
we obtain the standard packing inequality
\[
   |\cC|\cdot \min_{S\in\cS_{n,q}} |B_r(S)|
   \;\le\;
   |\cS_{n,q}|.
\]

In contrast to classical Hamming spaces, ball sizes in $\cS_{n,q}$ depend on
the center.
Thus the relevant quantity is the \emph{minimal} radius-$r$ ball size.
Recall (Lemma~\ref{lem:extreme-min-ball}) that the smallest balls are attained
at extreme centers and satisfy
\[
   \min_{S\in\cS_{n,q}} |B_r(S)|
   =
   \binom{r+q-1}{q-1}.
\]

\begin{theorem}[Sphere-packing bound]
\label{thm:sp-bound}
For all $n,q$ and all $d\ge 1$,
\[
   |\cC|
   \;\le\;
   \frac{\binom{n+q-1}{q-1}}{\binom{r+q-1}{q-1}},
   \qquad
   r=\Bigl\lfloor\frac{d-1}{2}\Bigr\rfloor.
\]
Equivalently, for deletion-correction radius $t$ (i.e., $d(\cC)\ge 2t+1$),
\[
   S_q(n,t)
   \;\le\;
   \frac{\binom{n+q-1}{q-1}}{\binom{t+q-1}{q-1}}.
\]
\end{theorem}

\begin{proof}
The disjointness of radius-$r$ balls around codewords gives
\[
   |\cC|\cdot \min_{S\in\cS_{n,q}} |B_r(S)|
   \;\le\;
   |\cS_{n,q}|.
\]
By Lemma~\ref{lem:extreme-min-ball},
\[
   \min_{S\in\cS_{n,q}} |B_r(S)|
   =
   \binom{r+q-1}{q-1}.
\]
Substituting this identity yields the first inequality. The formulation in
terms of $S_q(n,t)$ is obtained by setting $r=t$.
\end{proof}

Fix $q$ and $r$, and let $n\to\infty$.
Using
\[
   \binom{n+q-1}{q-1}
   =
   \frac{n^{q-1}}{(q-1)!}\,(1+o(1)),
\]
we obtain the asymptotic form
\[
   |\cC|
   \;\le\;
   \frac{n^{q-1}}{(q-1)!\,\binom{r+q-1}{q-1}}\,(1+o(1)).
\]

\subsection{Code--anticode bounds and boundary truncation}
\label{subsec:anticode-bounds}

Another approach to upper bounding multiset deletion codes is based on the
geometry of large balls via difference vectors.
This viewpoint motivates code--anticode type arguments in the lattice
representation of the multiset space.

For $S,T\in\cS_{n,q}$ with multiplicity vectors
$\mathbf{x}_S=(x_0,\dots,x_{q-1})$ and $\mathbf{x}_T=(y_0,\dots,y_{q-1})$, define
the difference vector
\[
   \mathbf{z}=\mathbf{x}_T-\mathbf{x}_S\in\Z^q,
\]
and write $\mathbf{z}=\mathbf{z}^+-\mathbf{z}^-$ with
$z_i^+=\max\{z_i,0\}$ and $z_i^-=\max\{-z_i,0\}$.
Then
\[
   d(S,T)=\sum_i z_i^+=\sum_i z_i^-,
\]
and feasibility is equivalent to the componentwise constraint
$\mathbf{z}^-\le \mathbf{x}_S$.

For a fixed center $S$, the radius-$r$ ball can be written in difference-vector
coordinates as
\[
   \Delta_r(S)
   \;\eqdef\;
   \{\mathbf{x}_T-\mathbf{x}_S:\ T\in\cS_{n,q},\ d(S,T)\le r\},
\]
so that $|B_r(S)|=|\Delta_r(S)|$.
Recall from Section~\ref{sec:pre} the universal difference-vector set
$S_{q-1}(r,r)$. For every $S\in\cS_{n,q}$ one has
$\Delta_r(S)\subseteq S_{q-1}(r,r)$, and hence
\[
   |B_r(S)|\le |S_{q-1}(r,r)|.
\]
By Corollary~\ref{cor:Sr-explicit},
\[
   |S_{q-1}(r,r)|
   =
   \sum_{j=0}^{\min(r,q-1)}
      \binom{q-1}{j}^2
      \binom{r-j+q-1}{q-1}.
\]

In the infinite lattice $\Z^q$ equipped with the metric $\frac12\|\cdot\|_1$,
the set $S_{q-1}(r,r)$ behaves as a genuine anticode: any two of its elements
are at distance at most $2r$.
Consequently, if a code has minimum distance at least $2r+1$, then no translate
of $S_{q-1}(r,r)$ can contain more than one codeword.
This is the standard translate-packing intuition behind code--anticode bounds,
and it explains why one might expect a bound of the form
\[
   |\cC|\le \frac{|\cS_{n,q}|}{|S_{q-1}(r,r)|}.
\]

However, this intuition fails in the simplex $\cS_{n,q}$ because difference
vectors must satisfy the feasibility constraint $\mathbf{z}^-\le\mathbf{x}_S$.
Near the boundary of the simplex, many vectors in $S_{q-1}(r,r)$ are infeasible,
and the corresponding translates are truncated.
Different truncated translates may intersect inside the simplex, even though
the untruncated translates in the ambient lattice are disjoint.
As a result, the naive code--anticode bound above is not valid in general.

This phenomenon is illustrated by the following explicit example.
Let $q=3$, $n=6$, and $t=2$, so that codes correcting $t$ deletions have minimum
distance at least $2t+1=5$.
Consider the code
\[
   \cC=\{(0,0,6),\ (0,6,0),\ (6,0,0)\}\subseteq\cS_{6,3}.
\]
Any two distinct codewords have disjoint supports, hence $|S\cap T|=0$ and
$d(S,T)=6$.
Thus $\cC$ is a valid multiset $2$-deletion-correcting code with $|\cC|=3$.

For these parameters,
\[
   |\cS_{6,3}|=\binom{8}{2}=28,
   \qquad
   |S_{2}(2,2)|=19,
\]
in agreement with Corollary~\ref{cor:avg-ball-specializations}.
The naive anticode ratio $|\cS_{6,3}|/|S_{2}(2,2)|=28/19<2$ would incorrectly
exclude the existence of such a code.
The contradiction arises because all three codewords lie on the boundary of the
simplex, where the corresponding sets $\Delta_t(S)$ are heavily truncated.

Kova\v{c}evi\'c and Tan resolve this issue by explicitly accounting for boundary
effects.
In their framework, a code correcting $t$ deletions has minimum distance at least
$2t+1$, and they define
\[
   \beta(t,q)=|S_{q-1}(t,t)|.
\]
Their Theorem~14 yields the unconditional upper bound
\[
   S_q(n,t)
   \;\le\;
   \frac{\binom{n+q-1}{q-1}}{\beta(t,q)}
   \;+\;
   qt
   \sum_{j=1}^{qt}
      \binom{n+q-1-j}{q-2}.
\]

Applying this bound to the example above, we have $t=2$ and hence
$\beta(2,3)=|S_2(2,2)|=19$.
Moreover,
\[
   \sum_{j=1}^{6}\binom{8-j}{1}
   =7+6+5+4+3+2
   =27.
\]
The bound becomes
\[
   S_3(6,2)
   \;\le\;
   \frac{28}{19}+3\cdot 2\cdot 27
   \;=\;
   \frac{3106}{19},
\]
which is strictly larger than $|\cS_{6,3}|=28$.
This is entirely consistent: the bound is valid but vacuous, reflecting the fact
that in this finite regime the boundary term dominates.

In contrast, the sphere-packing bound is governed by the minimal radius-$2$ ball
size $\binom{4}{2}=6$ and yields
\[
   |\cC|\le \frac{28}{6}<5,
\]
which is nontrivial and compatible with the existence of a size-$3$ code.

This comparison shows that while $S_{q-1}(r,r)$ accurately captures the interior
geometry of large balls and is useful for asymptotic analysis, boundary
truncation can render anticode-based bounds weak or even vacuous in finite
dimensions.
In such cases, sphere-packing bounds based on minimal balls can be strictly
stronger.

\subsection{Gilbert--Varshamov lower bound via average balls}
\label{subsec:gv-bound}

We now derive a lower bound using a greedy packing argument.
Unlike the previous upper bounds, the Gilbert--Varshamov approach is governed
by \emph{average} balls rather than extremal ones.
This viewpoint is closely related to the generalized sphere-packing
principles introduced in~\cite{FazeliVardyYaakobi2014}, which emphasize
mass-distribution and average-ball arguments in non-homogeneous metric spaces.
It may also be viewed as a fixed-alphabet specialization of the broader
Gilbert--Varshamov perspective developed for $L_1$ spaces in
~\cite{GoyalDaoKovacevicKiah2025}.

Define the average radius-$t$ ball size by
\[
   \overline{B_t}(n,q)
   \;\eqdef\;
   \frac{1}{|\cS_{n,q}|}\sum_{S\in\cS_{n,q}} |B_t(S)|.
\]

\begin{theorem}[Generalized Gilbert--Varshamov bound]
\label{thm:gv-bound}
For all $n\ge 1$, $q\ge 2$, and $t\ge 0$,
\[
   S_q(n,t)
   \;\ge\;
   \frac{|\cS_{n,q}|}{\overline{B_t}(n,q)}
   \;=\;
   \frac{\binom{n+q-1}{q-1}}{\overline{B_t}(n,q)}.
\]
\end{theorem}

By Theorem~\ref{thm:avg-ball-general}, the denominator is now available in
closed form:
\[
   \overline{B_t}(n,q)
   =
   \frac{1}{\binom{n+q-1}{q-1}}
   \sum_{m=0}^{t}
   c_{q,m}\binom{n-m+q-1}{q-1}.
\]
Hence the Gilbert--Varshamov bound can be written exactly as
\[
   S_q(n,t)
   \;\ge\;
   \frac{\binom{n+q-1}{q-1}^2}
        {\sum_{m=0}^{t} c_{q,m}\binom{n-m+q-1}{q-1}}.
\]

For fixed $q$ and $t$, Lemma~\ref{lem:avg-ball-ideal} gives the asymptotic form
\[
   \overline{B_t}(n,q)
   =
   |S_{q-1}(t,t)| - O\!\left(\frac{1}{n}\right),
   \qquad n\to\infty,
\]
and therefore
\[
   S_q(n,t)
   \;\ge\;
   \frac{\binom{n+q-1}{q-1}}{|S_{q-1}(t,t)|}\,(1-o(1)),
   \qquad n\to\infty.
\]

\medskip
It is instructive to compare this lower bound with known upper bounds.
Sphere-packing bounds are governed by the minimal radius-$t$ ball size,
which equals $\binom{t+q-1}{q-1}$, whereas the Gilbert--Varshamov bound is
governed by the average ball size.
This distinction mirrors the dichotomy highlighted in the generalized
sphere-packing framework of~\cite{FazeliVardyYaakobi2014} between extremal
and average geometric behavior.

For $q=2$, these quantities coincide up to constant factors, yielding
asymptotically matching upper and lower bounds.
For $q\ge3$, however, they differ by a constant multiplicative factor,
so a non-vanishing gap remains between the best known upper and lower bounds.

\medskip
To summarize, sphere-packing bounds reflect the geometry of minimal balls;
upper bounds based on difference vectors capture the structure of large balls
but are sensitive to boundary truncation; and Gilbert--Varshamov bounds are
governed by average ball sizes, for which exact finite-$n$ formulas are now
available and match the interior-difference-set scale asymptotically for fixed $q$ and $t$.


\section{Conclusion}
\label{sec:conclusion}

We developed two complementary perspectives on multiset deletion-correcting
codes. On the constructive side, we presented polynomial Sidon-type
constructions over finite fields, showing that in the regime $t<q$ one obtains
explicit multiset deletion-correcting codes with redundancy $t+O(1)$,
independent of the blocklength $n$. On the geometric side, we analyzed deletion
balls in the multiset simplex via difference vectors and a diagonal reduction of
the relevant generating functions, proved a global characterization of the
centers minimizing and maximizing deletion-ball size, and obtained exact
formulas for the pair enumerator \(A_{n,q}(m)\) and for the average ball size,
including closed forms in the binary and ternary cases.

These geometric results translate directly into volume-based bounds. Minimal
balls yield sphere-packing upper bounds; the ideal set $S_{q-1}(r,r)$ clarifies
both the promise and the limitations of code--anticode arguments under boundary
truncation; and exact average balls lead to Gilbert--Varshamov lower bounds that
match the interior-difference-set scale asymptotically for fixed $q$ and $t$.

Together with our previous work~\cite{CompanionPaper2026}, this paper provides
both explicit constructions and a geometric understanding of the multiset
deletion metric. Several directions remain open. It would be particularly
interesting to sharpen the gap between sphere-packing and average-ball bounds
for $q\ge 3$, to determine whether the polynomial approach can be extended beyond
the regime $t<q$, and to better understand when cyclic and polynomial
Sidon-based constructions are genuinely optimal in the multiset setting.

\end{document}